\documentclass[11pt,a4wide]{article}

\usepackage{amssymb}
\usepackage{amsfonts}
\usepackage{amsmath}
\usepackage{amsthm}
\usepackage[all]{xy}
\usepackage{latexsym}
\usepackage{enumerate}
\usepackage{algorithm}

\usepackage{dsfont}
\usepackage{bbm}

\usepackage{fancyhdr}

\DeclareMathOperator{\wdeg}{wdeg}

\newcommand{\bbN}{\mathbb{N}}
\newcommand{\bbNp}{\mathbb{N}^*}
\newcommand{\bbZ}{\mathbb{Z}}
\newcommand{\bbQ}{\mathbb{Q}}

\newcommand{\ef}{\mathbb{F}}

\newcommand{\efq}{\ef_q}
\newcommand{\eftwo}{\ef_2}
\newcommand{\eftwomm}[1]{\ef_{2^{#1}}}

\newcommand{\eftwos}{\eftwomm{s}}

\newcommand{\bs}{\boldsymbol}

\newcommand{\veps}{\varepsilon}
\newcommand{\tl}{\tilde}

\newcommand{\bch}{\mathrm{BCH}}

\newcommand{\lm}{\textsc{lm}}

\newcommand{\bx}{\bs{X}}
\newcommand{\by}{\bs{Y}}
\newcommand{\ba}{\bs{a}}
\newcommand{\bi}{\bs{i}}
\newcommand{\bj}{\bs{j}}
\newcommand{\bxbi}{\bx^{\bi}}
\newcommand{\bybi}{\by^{\bi}}
\newcommand{\bybj}{\by^{\bj}}
\newcommand{\hasse}[2]{{#1}^{(#2)}}
\newcommand{\hasseg}[1]{\hasse{g}{#1}}
\newcommand{\hassegi}{\hasseg{\bi}}
\newcommand{\hassegj}{\hasseg{\bj}}
\newcommand{\wt}{\mathrm{wt}}

\newcommand{\isom}{\stackrel{\sim}{\to}}

\newcommand{\roots}{\mathrm{Roots}}

\newcommand{\wdegpm}{\wdeg_{(1,-1)}}
\newcommand{\wdegw}{\wdeg_{(1,w)}}
\newcommand{\odd}[1]{\overset{\mathrm{odd}}{#1}}
\newcommand{\even}[1]{\overset{\mathrm{even}}{#1}}
\newcommand{\varared}[2]{\overset{(#2)}{#1}}
\newcommand{\ared}[1]{\varared{#1}{r}}
\newcommand{\armoed}[1]{\varared{#1}{r-1}}

\newcommand{\bg}{\bs{g}}
\newcommand{\bh}{\bs{h}}

\newcommand{\plus}[1]{{#1}^{+}}
\newcommand{\Gplus}{\plus{G}}
\newcommand{\gplus}{\plus{g}}
\newcommand{\bgplus}{\plus{\bg}}
\newcommand{\Mplus}{\plus{M}}

\newcommand{\tlftwomx}{L}
\newcommand{\ord}{\mathrm{ord}}

\newcommand{\wieven}[1]{w_{#1}^{\mathrm{even}}}
\newcommand{\wiodd}[1]{w_{#1}^{\mathrm{odd}}}
\newcommand{\woneeven}{\wieven{1}}
\newcommand{\woneodd}{\wiodd{1}}
\newcommand{\wtwoeven}{\wieven{2}}
\newcommand{\wtwoodd}{\wiodd{2}}

\newcommand{\homog}{\mathrm{Homog}_{Y,Z}}
\newcommand{\mult}{\mathrm{mult}}
\newcommand{\hath}{\hat{h}}

\newcommand{\bm}{\bs{m}}
\newcommand{\am}{\frak{a}(\bm)}

\newcommand{\nvars}[3]{N_{1,#1,#2}^{#3}}

\newcommand{\nvarsrwt}{\nvars{w'_1}{w'_2}{\rho}}

\newcommand{\mindeg}[3]{\Delta_{1,#1,#2}^{#3}}
\newcommand{\mindegrw}{\mindeg{w_1}{w_2}{\rho}}
\newcommand{\mindegrwt}{\mindeg{w'_1}{w'_2}{\rho}}

\newcommand{\eneq}{n_{\mathrm{eq}}}
\newcommand{\cost}{\mathrm{cost}}

\newcommand{\one}{\mathds{1}}

\newcommand{\naive}{na\"ive}
\newcommand{\kavcic}{Kav\v ci\' c} 

\newcommand{\downed}[1]{{#1}_{\downarrow}}
\newcommand{\predowned}[1]{{#1}^{\downarrow}}
\newcommand{\tlf}{\downed{f}}
\newcommand{\hatf}{\predowned{f}}

\newcommand{\lmvar}[1]{\lm_{<_{#1}}}
\newcommand{\lmw}{\lmvar{w}}
\newcommand{\lmmo}{\lmvar{-1}}

\newcommand{\pade}{Pad\'e}
\newcommand{\rmax}{r_{\max}}

\newcommand{\grobner}{Gr\"obner}
\newcommand{\kotter}{K\"otter}

\newtheorem{proposition}{Proposition}[section]
\newtheorem{definition}[proposition]{Definition}
\newtheorem{corollary}[proposition]{Corollary}

\newtheorem{theorem}[proposition]{Theorem}
\newtheorem{remark}[proposition]{Remark}

\newtheorem*{remarknn}{Remark}

\newtheorem*{examplenn}{Example}

\title{Fast syndrome-based Chase decoding of binary BCH codes through
Wu list decoding}  

\date{\today} 
\author{Yaron Shany and Amit Berman\thanks{The
authors are with Samsung  Semiconductor Israel R\&D Center, 146 Derech
Menachem Begin St., Tel Aviv, 6492103, Israel. Emails: \{yaron.shany,
amit.berman\}@samsung.com}}

\begin{document}
\maketitle

\begin{abstract}
We present a new fast Chase decoding algorithm for binary BCH codes.
The new algorithm reduces the
complexity in comparison to a recent fast Chase decoding algorithm
for Reed--Solomon (RS) codes by the authors (IEEE Trans.~IT, 2022), by
requiring only a single \kotter{} iteration per edge of the 
decoding tree. In comparison to  
the fast Chase algorithms presented by Kamiya (IEEE Trans.~IT, 2001)
and Wu (IEEE Trans.~IT, 2012) for binary BCH codes, the polynomials
updated throughout the algorithm of the current paper typically have a
much lower degree. 

To achieve the complexity reduction, we build on a new
isomorphism between two solution modules in the binary case, and on a
degenerate case of the soft-decision (SD) version of the Wu list
decoding algorithm. Roughly speaking, we prove that when the maximum
list size is $1$ in Wu list decoding of binary BCH codes, assigning a
multiplicity of $1$ to a coordinate has the same effect as flipping
this coordinate in a Chase-decoding trial. 
  
The solution-module isomorphism also provides a systematic
way to benefit from the binary alphabet for reducing the complexity in
bounded-distance hard-decision (HD) decoding. Along the way, we 
briefly develop the \grobner{}-bases formulation of the Wu list
decoding algorithm for binary BCH codes, which is missing in the
literature. 
\end{abstract}

\section{Introduction}

\subsection{Motivation and known results}
Binary BCH codes are widely used in storage and communication
systems. Traditionally, these codes are decoded by {\it
hard-decision (HD)} decoding algorithms that perform unique decoding up
to half the minimum distance, such as the Berlekamp--Massey (BM)
algorithm \cite{Ma69}. 

The revolutionary work \cite{GS99} of Guruswami and Sudan (following
Sudan's original work \cite{S97}) introduced a polynomial time list
decoding algorithm for Reed--Solomon (RS) codes over $\efq$ (where $q$
is a prime power and $\efq$ is the finite field of $q$ elements) up to
the so-called $q$-ary  {\it Johnson bound} \cite[Sec.~9.6]{Roth},
\cite[p.~3276]{BHNW13}. Wu \cite{Wu08}  
introduced an even more efficient algorithm for HD list decoding of RS
codes up to the Johnson bound. 

Since BCH codes are subfield subcodes of RS codes with the same
minimum distance, every algorithm for decoding RS codes is also
automatically an algorithm for decoding BCH codes. Hence, the list
decoding algorithms for RS codes from \cite{GS99} and \cite{Wu08} can
be used for decoding BCH codes up to the $q$-ary Johnson radius.
However, if the BCH code is defined over a proper subfield $\ef_{q'}$
of $\efq$, then one may expect a larger decoding radius, since the
$q'$-ary Johnson bound is in general larger than the $q$-ary Johnson
bound for $q'<q$. For $q'=2$, \cite{Wu08} includes also
a list decoding algorithm for binary BCH codes that takes advantage of
the binary alphabet, and reaches the binary Johnson bound. 

It is well known that using {\it channel reliability information} and
moving from HD decoding to {\it soft-decision (SD)} decoding may
significantly improve the decoding performance. \kotter{} and Vardy
\cite{KV03} presented a polynomial time SD list decoding algorithm for RS
codes that converts the channel reliability information into a {\it
multiplicity matrix}, with one dimension corresponding to the
coordinates of the transmitted codeword, and the other to the
 $q$ possible symbols of the alphabet. 

In the conclusion of \cite{Wu08}, Wu mentions in passing that his list
decoding algorithms for RS and BCH codes can be converted to SD
list decoding algorithms, by using  
different multiplicities for different coordinates, depending on
the channel reliability information. While Wu's ``one-dimensional''
multiplicity assignment is not as general as the two-dimensional
multiplicity assignment of \cite{KV03} for $q>2$, this is not the
situation in the binary case.  

Before \cite{KV03} and \cite{Wu08}, the main algebraic SD decoding
algorithms for binary BCH codes were the generalized minimum distance
(GMD) decoding \cite{F66}, and the Chase decoding algorithms
\cite{Chase}. GMD decoding consists of successively erasing an even
number of the least reliable coordinates and applying
errors-and-erasures decoding.  
In Chase decoding, there is a predefined list of test error
patterns on the $\eta$ least reliable coordinates for some small
$\eta$. For example, this list may consist of a random list of
vectors, all possible non-zero
vectors, all vectors of a low enough weight,\footnote{Throughout,
``weight'' and ``distance'' refer to Hamming weight and Hamming
distance, respectively.} 
etc.. The decoder runs on error patterns from the list and
subtracts them from the received word, feeding the result to an HD
decoder. If the HD decoder succeeds, then its output is saved into the
output list of the decoder.

While the Chase algorithms typically have exponential complexity, they
are known to have better performance than the algorithms of
\cite{KV03} and \cite{Wu08} for high-rate short to medium length codes
\cite{Wu12}. Therefore, there is still a great interest in finding
low-complexity Chase decoding algorithms. 

In {\it fast} Chase decoding algorithms, the decoder shares
computations between HD decodings of different test error
patterns (see, e.g., \cite{BK10}, \cite{Wu12}, \cite{ZZW09},
\cite{XCB20}, \cite{SB21}, \cite{K01}, \cite{Wu12},
and \cite{Z13}). Several papers used special properties of
\emph{binary} BCH codes  
to further reduce the complexity of their Chase decoding
algorithms. Kamiya \cite{K01} presented a fast Chase decoding
algorithm for binary BCH codes building on the Welch--Berlekamp
algorithm. Wu \cite[Appendix C]{Wu12} presented a simplified version
of his BM-based polynomial-update algorithm for binary BCH codes. 
Zhang {\it et al.}~\cite{ZZW11} extend \cite{Wu12} by  
introducing  a ``backward'' step, enabling to order the test vectors
according to a Gray map instead of using Wu's tree. Finally, Zhang
\cite{Z13} suggested an optimized version of \cite{BK10}, using both
optimizations from \cite{ZZ10} and \cite{ZZ12} (for RS codes in
general), and the binary alphabet of the BCH codes. 

It should be noted that \cite{BK10}, \cite{Z13} are
``time domain'' algorithms,  that is, they work directly on the
received vector, and not on the {\it syndrome} vector. For
practical applications using high-rate BCH codes, it is typically
beneficial to replace the long received vector by the short syndrome
vector before the decoding begins. 

Because our fast Chase algorithm is based on a \grobner{} bases
formulation of the Wu list decoding algorithm, let
us recall the main works in this line of research. For RS
codes, the \grobner{} bases formulation of the Wu list decoding algorithm
was introduced by Trifonov \cite{Trif10}, using a time 
domain approach. Beelen {\it et al.} \cite{BHNW13} presented a
\grobner{} bases formulation working with the syndrome
vector. Irreducible binary Goppa codes were also considered in
\cite{BHNW13}, where an algorithm for HD list decoding up to the binary
Johnson radius was developed.

\subsection{Results and method}
We present a syndrome-based fast Chase decoding algorithm for
(primitive, narrow-sense) binary
BCH codes that considerably reduces the complexity in comparison to 
\cite[Alg.~C]{SB21} (reduction by a factor of about $5$ for polynomial
updates; see Subsection \ref{subsec:complexity} for a precise
statement). To do that, we use a different approach than that of
\cite{SB21}.  

First, we establish an isomorphism between two solution
modules for binary BCH codes, one for the {\it key equation} and one
for a {\it modified key equation}. This isomorphism gives a general
framework for reducing the decoding complexity in the binary case,
both in HD bounded-distance decoding, and in fast Chase decoding. For
example, for HD bounded-distance decoding, this isomorphism enables to
benefit from the binary alphabet with practically any existing
algorithm (such as the Euclidean algorithm), without being tied only
to Berlekamp's well-known simplification of his algorithm
\cite[pp.~24--32]{Berl66}. We remark 
that a similar idea exists in the literature for specific
algorithms \cite{Ca88}, \cite{JK95}. Also, \cite{JK95} introduces a
general method with a different approach (using Newton's identities)
and with a distinction between even and odd correction radius.

Building on the above isomorphism, we then use an SD version of the Wu
list decoding algorithm to derive a 
syndrome-based fast Chase decoding algorithm for binary BCH
codes. In the new algorithm there is just one
\kotter{} iteration per edge of the decoding tree (see ahead for
details), as opposed to two iterations in \cite{SB21} for RS
codes. Halving the number of \kotter{} iterations 
has a double effect in reducing the complexity: once in reducing the
degrees of the maintained polynomials, and once in performing less
substitutions and multiplications by scalars. 

We note that it is not possible to simply omit an iteration from the
algorithm of \cite{SB21} in the binary case: omitting the {\it derivative
iteration} from the algorithm of \cite{SB21} makes it into
a fast GMD algorithm, where coordinates are dynamically \emph{erased},
which is different from dynamically \emph{flipping} bits as in Chase
decoding. 

Similarly to \cite[Alg.~C]{SB21}, instead of tracking the
ELP itself, we track low-degree ``coefficient polynomials'' in the
decomposition of the ELP according to a suitable basis for the
solution module of the modified key equation. Consequently, if the unique
error-correction radius is $t$, the total number of errors is $t+r$,
and the total number of errors in the non-reliable coordinates is at
least $r$, then the maximum sum of degrees of all maintained
polynomials is about $2r$ (for a precise statement, see Proposition
\ref{prop:degs} ahead), as opposed to about $2(t+r)$ in the algorithm
of Wu \cite{Wu12}.   

To benefit from the low degree of the updated polynomials, we avoid
working with evaluation vectors.  
While working with evaluation vectors is a theoretical means for
achieving a complexity of $O(n)$ per modified coordinate, it is
inefficient for practical code parameters. For this reason, Wu
\cite[Sec.~5]{Wu12} suggests a criterion for entering the evaluation
step, based on a length variable of his algorithm. As in \cite{SB21},
we use such a 
criterion which is based on a {\it discrepancy} of the algorithm: ELP
evaluation is performed when the discrepancy is zero. This criterion
never misses a required evaluation, but might 
invoke an unnecessary evaluation with a low probability, and also
slightly degrades the probability of decoding success: When the total
number of errors is $t+r$, we require $r+1$ (instead of $r$) errors on
the unreliable coordinates, similarly to \cite{Wu12}.

Because our algorithm tracks low-degree coefficient polynomials,
evaluating an ELP hypothesis amounts to evaluating two polynomials of
degree about $r$. Hence the number of finite-field multiplications of the
evaluation step (when the above criterion is satisfied) is
$O(n\cdot \min\{r,\log(n)\})$, where the minimum is between \naive\ 
point-by-point evaluation, and the fast additive FFT 
algorithm of Gao and Mateer \cite{GM10}.  

We use the language of \grobner{}-bases for $\efq[X]$-modules, following
the works of Fitzpatrick \cite{Fitz95}, Trifonov \cite{Trif10},
Trifonov and Lee \cite{TL12}, and Beelen {\it et
al.}~\cite{BHNW13}. Along the way, we
briefly derive the \grobner{}-bases formulation of the (SD) Wu list
decoding algorithm for binary BCH codes, which is missing in the
literature. While quite similar to the 
case of binary irreducible Goppa codes considered in \cite{BHNW13},
there are some differences that must be considered. For example, the
method of \cite{BHNW13} builds on the fact that for an irreducible
$g(X)\in \eftwos[X]$ (for $s\in \bbNp$), 
$\eftwos[X]/(g(X))$ is a field, so that any non-zero element  has a
multiplicative inverse. However,  
this is not the case for the ring $\eftwos[X]/(X^{d-1})$ involved in
the key equation for binary BCH codes.

In summary, the main contributions of the paper are as follows:

\begin{itemize}

\item We introduce a new fast Chase algorithm for binary BCH codes,
namely, 
Algorithm A, that requires a single \kotter{} iteration per edge of the
decoding tree, and reduces the complexity by a factor of about $5$
in comparison to the corresponding algorithm for RS codes
\cite{SB21}.\footnote{Again, we note that it is not possible to simply
omit an iteration from the algorithm 
of \cite{SB21}.} In comparison to existing fast Chase algorithms for
binary BCH codes, the new algorithm maintains polynomials of a lower
degree, and consequently the complexity is considerably lower; see
Subsection \ref{subsec:complexity} for details.

\item We use a degenerate version of the SD Wu list decoding algorithm for
devising the new fast Chase algorithm, thus 
establishing a connection between two prominent algebraic SD decoding
algorithms for binary BCH codes. Roughly speaking, the idea is that
when the list size is $1$ in the Wu list decoding algorithm, assigning
a multiplicity of $1$ to a 
coordinate has the same effect as flipping the coordinate in a Chase 
trial (Theorem \ref{thm:success}). We believe that this connection is
more natural than that of \cite{BK10}, where the 
Guruswami--Sudan list decoding algorithm was used to devise a fast
Chase algorithm for RS codes. 

\item We establish an isomorphism between two solution modules for
decoding binary BCH codes. This isomorphism provides a general
framework for benefiting in complexity from the binary alphabet, even
in bounded distance decoding,
without being tide only to Berlekamp's simplification of his algorithm
\cite{Berl66}. 

\item We derive the \grobner{}-bases formulation of the SD Wu list
decoding algorithm for binary BCH codes, which is missing in the
literature. 

\end{itemize}

\subsection{Organization}
Section \ref{sec:preliminaries} contains some definitions and results
used throughout the paper. In Section \ref{sec:isom}, we prove the
isomorphism between two solution modules for decoding binary BCH
codes, and show how this isomorphism can be used for benefiting from
the binary alphabet in a systematic
way. Section \ref{sec:softwu} contains a brief derivation of the
\grobner{}-bases approach for SD Wu list decoding of binary BCH
codes. 
The new fast Chase decoding algorithm is derived in Section
\ref{sec:chase}, in which we describe the update rule on an edge
of the decoding tree, the stopping criterion for entering
exhaustive evaluation, and efficient methods for preforming the
evaluation. The same section also includes a detailed complexity
analysis. Finally, Section \ref{sec:conclusions} includes some
conclusions and open questions for further research.

\section{Preliminaries}\label{sec:preliminaries}

\subsection{Binary BCH codes and the key equation}

For $s,t\in\bbNp$, let $\bch(t)$ be the primitive binary BCH
code of length $2^s-1$ and designed distance $d:=2t+1$, that is, let
$\bch(t)$ be the cyclic binary code of length $n:=2^s-1$ with zeroes
$\gamma,\gamma^3,\ldots,\gamma^{2t-1}$ and their conjugates, where
$\gamma$ is some primitive element in $\eftwos^*$. 

Suppose that a codeword $\bs{x}\in \bch(t)$ is transmitted, and the  
HD received word is $\bs{y}:=\bs{x}+\bs{e}$ for some
error vector $\bs{e}\in \eftwo^n$.  For $j\in\bbN$, let
$S_j:=y(\gamma^j)$, where for a field $K$ and a vector
$\bs{f}=(f_0,f_1,\ldots,f_{n-1})\in K^n$, we let
$f(X):=f_0+f_1X+\cdots + f_{n-1}X^{n-1}\in 
K[X]$. The {\bf syndrome polynomial}
associated with $\bs{y}$ is
$S(X):=S_1+S_2X+\cdots+S_{d-1}X^{d-2}$. 

Suppose that the error locations are some distinct elements
$\alpha_1,\ldots,\alpha_{\veps}\in \eftwos^*$ (where 
$\veps\in\{1,\ldots,n\}$ is the number of errors), and let
the {\bf error locator polynomial (ELP)}, $\sigma(X)\in\eftwos[X]$, be
defined by
$$
\sigma(X):=\prod_{i=1}^{\veps}(1+\alpha_iX), 
$$
and the {\bf error evaluator polynomial (EEP)}, $\omega(X)\in\eftwos[X]$,
be defined by
$$
\omega(X):=\sum_{i=1}^{\veps} \alpha_i\prod_{j\neq
i}(1+\alpha_jX)=\sigma'(X),
$$
where $\sigma'(X)$ is the formal derivative of $\sigma(X)$.
Then $\omega$ and $\sigma$ satisfy the following well-known {\bf key equation}:
\begin{equation}\label{eq:key}
\omega\equiv S\sigma \mod (X^{2t})
\end{equation}
(see, e.g., \cite[Sect.~6.3]{Roth} for a proof).  

\subsection{A useful monomial ordering}
For monomial orderings and \grobner{} bases for modules over polynomial 
rings, we refer to \cite[Sec.~5.2]{CLO2}. In this section, we recall
a monomial ordering used throughout the paper.  

\begin{definition}\label{def:ord}
{\rm

Let $K$ be an arbitrary field.

\begin{enumerate}

\item For $\bs{w}=(w_1,\ldots,w_n)\in \bbQ^{n}$ (where
$n\in \bbNp$) and for $0\neq f(X_1,\ldots,X_n)\in
K[X_1,\ldots, X_n]$, we let $\wdeg_{\bs{w}}(f)$ be the maximum value of
$w_1i_1+w_2i_2+\cdots +w_{n}i_{n}$ over all
monomials $X_1^{i_1}X_2^{i_2}\cdots X_{n}^{i_n}$ appearing in
$f$. For $f=0$, we put $\wdeg_{\bs{w}}(f):=-\infty$. We call
$\wdeg_{\bs{w}}(f)$ the {\bf $\bs{w}$-weighted degree} of 
$f$.

\item For $w\in \bbQ$, and for $f_0(X),f_1(X)\in K[X]$, we define 
$\wdegw(f_0,f_1):=\wdegw(f_0(X)+Yf_1(X))$. In other words,
$$
\wdegw(f_0,f_1):=\max\{\deg(f_0),\deg(f_1)+w\}.
$$ 

\item Using the notation of \cite{Roth}, put
$\ord(f_0,f_1):=\max\{\deg(f_0)+1,\deg(f_1)\}$, and note that
$\ord(f_0,f_1)=\wdegpm(f_0,f_1)+1$. As in \cite{Roth}, we refer to
$\ord(f_0,f_1)$ as the {\bf order} of the pair
$(f_0,f_1)$.  

\item Following \cite{Fitz95}, define a monomial ordering $<_w$ on
$K[X]^2$ as follows. For monomials\footnote{Recall that a monomial in
$K[X]^{2}$ is a pair of the form $(X^j,0)$ or $(0,X^j)$ for some $i\in
\bbN$.} $\bs{m}_1,\bs{m}_2\in K[X]^2$,
define $\bs{m}_1<_w\bs{m}_2$ if $\wdegw(\bs{m}_1)<\wdegw(\bs{m}_2)$,
or $\wdegw(\bs{m}_1)=\wdegw(\bs{m}_2)$, and, when a pair $(f(X),g(X))$
is considered as the bivariate polynomial $f(X)+Yg(X)$, $\bs{m}_1$ is
smaller than $\bs{m}_2$ with respect to the lex ordering with
$Y>X$. Explicitly, we define $(X^{j_1},0)<_w(X^{j_2},0)$ iff $j_1<j_2$,
$(0,X^{j_1})<_w(0,X^{j_2})$ iff $j_1<j_2$, and
$(X^{j_1},0)<_w(0,X^{j_2})$ iff $j_1\leq j_2+w$.  Note that this is a
monomial ordering even when $w$ is not positive or not an integer (as
in \cite{BHNW13}). 

\end{enumerate}
}
\end{definition}

In what follows, for a monomial ordering $\prec$ on $K[X]^{\ell}$ (for
$\ell\in \bbNp$), and for
$\bs{f}\in K[X]^{\ell}$, we will write $\lm_{\prec}(\bs{f})$ for the
leading monomial of $\bs{f}$  with respect to $\prec$, that is, for
the $\prec$-largest monomial appearing in $\bs{f}$. 
When $\prec$ is clear from the context, we will sometimes write
$\lm(\bs{f})$ for $\lm_{\prec}(\bs{f})$. 
Finally, for monomials in $K[X]^2$, we say that the monomial $(X^i,0)$
is {\bf on the left}, while the monomial $(0,X^i)$ is on {\bf on the
right}.   

\subsection{\kotter{}'s iteration}
In this subsection we recall the general form of \kotter{}'s
iteration \cite{Koetter}, \cite{NH98}, as presented
by McEliece \cite[Sect.~VII.C]{McE03}. Throughout this subsection, $K$
is an arbitrary field.

For  a $K[X]$-submodule $M\subseteq K[X]^{\ell+1}$ (for $\ell\in
\bbNp$) of rank $\ell+1$, let
$G=\{\bs{g}_0,\ldots,\bs{g}_{\ell}\}$ be a \grobner{} basis for $M$ with 
respect to some monomial ordering $\prec$ on $K[X]^{\ell+1}$. We will
assume that the leading monomial of $\bs{g}_j$ contains the $j$-th unit
vector, for all $j\in\{0,\ldots, \ell\}$, where coordinates of vectors
are indexed by $0,\ldots, \ell$ (it is easily verified that there is
no loss of generality in this assumption).  

Let $D\colon K[X]^{\ell+1}\to K$ be a non-zero linear functional
for which $\Mplus:= M\cap \ker(D)$ is a $K[X]$-module.
\kotter{}'s iteration converts the $(\ell+1)$-element
\grobner{} basis\footnote{From this point on in this subsection,
``\grobner{} basis'' means a \grobner{} basis with respect to $\prec$.} 
$G$ of $M$ to an $(\ell+1)$-element \grobner{} basis
$\Gplus=\{\bgplus_0,\ldots,\bgplus_{\ell}\}$ of $\Mplus$, while 
maintaining the property that $\lm(\bgplus_j)$ contains the 
$j$-th unit vector for all $j\in\{0,\ldots,\ell\}$. 
It is presented in the following pseudo-code.
\hfill\\ 

\begin{tabular}{|c|}
\hline
{\bf \kotter{}'s iteration}\\
\hline
\end{tabular}

\begin{description}

\item[\bf{Input}] A \grobner{} basis
$G=\{\bs{g}_0,\ldots,\bs{g}_{\ell}\}$ for the submodule $M\subseteq
K[X]^{\ell+1}$, with $\lm(\bs{g}_j)$ containing the $j$-th unit vector for
all $j$ 

\item[\bf{Output}] A \grobner{} basis
$\Gplus=\{\bgplus_0,\ldots,\bgplus_{\ell}\}$ for $\Mplus$ with
$\lm(\bgplus_j)$ containing the $j$-th unit vector for all $j$

\item[\bf{Algorithm}]

\begin{itemize} 

\item For $j=0, \ldots, \ell$, calculate $\Delta_j:=D(\bs{g}_j)$

\item Set $J:=\big\{j\in\{0,\ldots,\ell\}|\Delta_j\neq 0\big\}$

\item For $j \in\{0, \ldots, \ell\}\smallsetminus J$,

\begin{itemize}

\item  Set $\bgplus_j:=\bs{g}_j$

\end{itemize}

\item Let $j^*\in J$ be such that  $\lm(\bs{g}_{j^*})=\min_{j\in
J}\{\lm(\bs{g}_j)\}$ 

\item For $j\in J$ 

\begin{itemize}

\item If $j\neq j^*$

\begin{itemize}

\item Set $\bgplus_j := \bs{g}_j-\frac{\Delta_j}{\Delta_{j^*}}
\bs{g}_{j^*}$  

\end{itemize}

\item Else /* $j=j^*$ */

\begin{itemize}

\item Set $\bgplus_{j^*} :=  X \bs{g}_{j^*}  -
\frac{D(X\bs{g}_{j^*})}{\Delta_{j^*}}\cdot \bs{g}_{j^*}$\\ 
/* $=\Big(X - \frac{D(X\bs{g}_{j^*})}{\Delta_{j^*}}\Big)\bs{g}_{j^*}$
*/ 

\end{itemize}

\end{itemize}

\end{itemize}

\end{description}

Note that the introduction of the additional set of variables
$\{\bs{g}_j^+\}$ is only for clarity of presentation. 
For a proof that at the end of \kotter{}'s iteration it indeed holds
that $\Gplus=\{\bgplus_0,\ldots,\bgplus_{\ell}\}$ is a \grobner{} basis
for $\Mplus$ and for all $j$, $\lm(\bgplus_j)$ contains the $j$-th
unit vector, see \cite[Sec.~VII.C]{McE03}.

\begin{remarknn}
{\rm
Algorithms similar to \kotter{}'s iteration were
already presented in \cite{BL94}, \cite{BL97},
which came before \cite{Koetter} (see \cite[Sec.~2.6]{JNSV17}). For some
problems over univariate polynomial rings, algorithms for computing
the (weak or canonical) shifted {\it Popov form} of $K[X]$-matrices
have the lowest asymptotic complexity; see, e.g., \cite{N16},
\cite{RS21} and the references therein. 
We currently do not know if such methods can be used to further reduce
the complexity in the context of the present paper. 
}
\end{remarknn}

\subsection{Hasse derivatives and multiplicities}
In this subsection, we recall some definitions and results
regarding Hasse derivatives and multiplicities of polynomials 
in $K[X_1,\ldots, X_n]$ for $K$ a field and $n\in \bbNp$. 

In what follows, we let $\bx:=(X_1,\ldots,X_n)$ and
$\by:=(Y_1,\ldots,Y_n)$, where $X_1,\ldots,X_n,Y_1,\ldots,Y_n$ are
algebraically independent over $K$. 
For $g(\bx)\in K[\bx]$ and $\bi=(i_1,\ldots,i_n)\in\bbN^n$, the
$\bi$-th {\bf Hasse derivative} of $g$ is defined as the unique coefficient
$\hassegi(\bx)$ of $\bybi:=Y_1^{i_1}\cdots Y_n^{i_n}$ in the
decomposition $g(\bx+\by)=\sum_{\bj} \hassegj(\bx)\bybj$. 

For binomial coefficients, we will use the convention that
$\binom{k}{\ell}=0$ if $\ell>k$.  We omit the straightforward proof of
the following well-known proposition.

\begin{proposition}\label{cw_prop:hassecalc}
If $g(\bx)=\sum_{\bi} a_{\bi}\bxbi$
with $a_{\bi}\in K$ for all $\bi$ and with $a_{\bi}=0$ for almost all
$\bi$, then for all $\bj$,  
\begin{equation}\label{cw_eq:hasse_sum}
\hasseg{\bj}(\bx) = \sum_{\bi}
a_{\bi}\binom{i_1}{j_1}\cdots\binom{i_n}{j_n}\bx^{\bi-\bj} 
\end{equation}
(note that all terms with a non-zero coefficient involve only
non-negative exponents).
\end{proposition}

For $\bi=(i_1,\ldots,i_n)$, let $\wt(\bi):=\sum_{j=1}^n i_j$. For
non-zero $g\in K[\bx]$ and for $\ba\in K^n$, we define the {\bf 
multiplicity} $\mult(g,\ba)$ of $g$ at $\ba$ to be the unique integer
$m$ for which for all $\bi=(i_1,\ldots,i_n)$ with $\wt(\bi) < m$,
$\hassegi(\ba)=0$, while there exists an $\bi$ with $\wt(\bi) = m$ for
which $\hassegi(\ba)\neq 0$.  If
$g$ is the zero polynomial, then we define $\mult(g,\ba):=\infty$ for
all $\ba$.  

Note that by the definition of the Hasse derivative, $\hassegi(\ba)$ is
just the coefficient of $\bxbi$ in the decomposition
$g(\ba+\bx)=\sum_{\bj}c_{\bj}\bx^{\bj}$ with $c_{\bj}\in K$ for all
$\bj$. Hence, $\mult(g,\ba)$ is simply the smallest value of
$\wt(i_1,\ldots,i_n)$ for a monomial $X_1^{i_1}\cdots X_n^{i_n}$
appearing in $g(\bx+\ba)$.   

The following proposition, involving polynomials with a (possibly)
different number of variables, is a special case of \cite[Prop.~6]{DKSS09}.  

\begin{proposition}{\cite{DKSS09}}\label{cw_prop:compmult}
Let $g\in K[X_1,\ldots,X_n]$, and let $h_1,\ldots,h_n\in K[Y_1,\ldots,
Y_r]$. Then for all $\ba\in K^r$,
\begin{equation}\label{cw_eq:compmult}
\mult\Big(g\big(h_1(\by),\ldots,h_n(\by)\big),\ba\Big)\geq 
\mult\Big(g,\big(h_1(\ba),\ldots,h_n(\ba)\big)\Big).
\end{equation}
\end{proposition}

\subsection{Partially homogenized polynomials}
Trifonov \cite{Trif10} introduced partially homogenized trivariate
polynomials as a means for avoiding points at infinity in Wu's
list decoding algorithm \cite{Wu08}. While in this paper we eventually
do (implicitly) use points at infinity, it is still convenient to work
with partially homogenized polynomials.

For a field $K$, a $(Y,Z)$-{\bf homogeneous polynomial} is a polynomial of the form
$$
f(X,Y,Z):=\sum_{j=0}^{\rho}\sum_{i}a_{ij}X^iY^jZ^{\rho-j}\in K[X,Y,Z]
$$ 
for some $\rho\in \bbN$, where $a_{ij}\in K$ for all $i,j$. 
We will say that $f$ has $(Y,Z)$-{\bf homogeneous degree} $\rho$, and write
$\homog^{K}(\rho)$ for the $K[X]$-module of polynomials
in $K[X,Y,Z]$ which are 
$(Y,Z)$-homogeneous with $(Y,Z)$-homogeneous degree $\rho$. We also
write $\homog^K:=\bigcup_{\rho\in \bbN^+} \homog^K(\rho)$. When the
underlying field $K$ is clear from the context, we will write simply
$\homog(\rho)$ and $\homog$ for $\homog^K(\rho)$ and $\homog^K$,
respectively.    

\begin{definition}\label{cw_def:tlf}
{\rm 
For $f(X,Y,Z):=\sum_{j=0}^{\rho}\sum_{i}a_{ij}X^iY^jZ^{\rho-j}$, where
$\rho\in\bbN$ and the $a_{ij}$ are in $K$, define the {\bf associated
(de-homogenized) bivariate polynomial of $f$}, $\tlf(X,Y)$, by setting
$$
\tlf(X,Y):=\sum_{j=0}^{\rho}\sum_i a_{ij}X^i Y^j=f(X,Y,1).
$$
Also, let $\hatf(X,Y)$ by the $Y$-reversed version of $\tlf(X,Y)$,
that is, 
$$
\hatf(X,Y):=Y^{\rho}\tlf(X,Y^{-1})=f(X,1,Y).
$$
}
\end{definition}

Note that $f$ is determined by $\tlf$ by homogenization, as
$f(X,Y,Z)=Z^{\rho}\tlf(X,Y/Z)$. Hence $f\mapsto \tlf$ is a
bijection between $(Y,Z)$-homogeneous trivariate polynomials of
$(Y,Z)$-homogeneous degree $\rho$ and bivariate polynomials of
$Y$-degree at most $\rho$. In fact, this is an isomorphism of
$K[X]$-modules, as $f\mapsto \tlf$ is clearly $K[X]$-linear.

The following proposition from \cite{Trif10} assures that
multiplicities are preserved when moving from $f$ to $\tlf$ (or
$\hatf$), where below we take special care of points of the form
$(x_0,0,0)$. 

\begin{proposition}{\cite[Lemma 1]{Trif10}}\label{cw_prop:yzhomog}
Let
$$
f(X,Y,Z):=\sum_{j=0}^{\rho}\sum_{i}a_{ij}X^iY^jZ^{\rho-j}\in\homog(\rho),
$$
and let $(x_0,y_0,z_0)\in K^3$. Then
\begin{enumerate}

\item For all $\alpha\in K^*$, $\mult\big(f,(x_0,y_0,z_0)\big) =
\mult\big(f,(x_0,\alpha y_0,\alpha z_0)\big)$.

\item If $z_0\neq 0$, then
$\mult\big(f,(x_0,y_0,z_0)\big)=\mult\big(\tlf,(x_0,y_0/z_0)\big)$. 

\item If $y_0\neq 0$, then
$\mult\big(f,(x_0,y_0,z_0)\big)=\mult\big(\hatf,(x_0,z_0/y_0)\big)$.

\item $\mult\big(f,(x_0,0,0)\big)\geq \mult\big(f,(x_0,y_0,z_0)\big)$.

\end{enumerate}
\end{proposition}

Parts 1--3 were proved in \cite{Trif10}. For a proof of part 4, see
Appendix \ref{app:homog}.

\section{An isomorphism between two solution modules}\label{sec:isom}
The solution set of the key equation $u'\equiv Su\mod(X^{2t})$ is an 
$\eftwos$-vector subspace of $\eftwos[X]$ that is not an
$\eftwos[X]$-submodule (i.e., an ideal). Also, to work with 
\grobner{} bases, it will be convenient to work with a submodule of
$\eftwos[X]^2$. In this section, we will see that with an appropriate scalar
polynomial multiplication, the key equation does define a module, and
this module is isomorphic to a sub-module of $\eftwos[X]^2$ which will be
useful for a \grobner{}-bases formulation.  
\begin{definition}
{\rm
For a field $K$, an even\footnote{There is no loss in generality in
the requirement that $s$ is even.} integer $s\in \bbN$, and a
polynomial $f(X)=f_0+f_1X+\cdots+f_s X^s\in K[X]$, the {\bf
odd part} of $f(X)$, $\odd{f}(X)$, is  
$$
\odd{f}(X):=f_1+f_3X+\cdots+f_{s-1}X^{s/2-1},
$$
while the {\bf even part} of $f(X)$, $\even{f}(X)$, is  
$$
\even{f}(X):=f_0+f_2X+\cdots+f_sX^{s/2}.
$$
}
\end{definition}

In the above definition, note that $f(X)=\even{f}(X^2)+X\odd{f}(X^2)$.

\begin{definition}\label{def:mu}
{\rm

\begin{enumerate}

\item Let us define an $\eftwos[X]$-module, $\tlftwomx$, as
follows. As an abelian group, $\tlftwomx=\eftwos[X]$, but the scalar
multiplication $\cdot_1\colon\eftwos[X]\times \tlftwomx\to L$ is
defined by $f(X)\cdot_1 g(X):=f(X^2)g(X)$.

\item Let 
\begin{eqnarray*}
\mu\colon \eftwos[X]\times\eftwos[X] & \to & \tlftwomx\\
(u(X),v(X)) & \mapsto & v(X^2)+Xu(X^2),
\end{eqnarray*}
and note that $\mu$ is an isomorphism of
$\eftwos[X]$-modules, with inverse
$f\mapsto(\odd{f},\even{f})$.\footnote{Characteristic $2$ is  
\emph{not} required for making $\tlftwomx$ a module and for showing
that $\mu$ is a homomorphism.} 

\end{enumerate}
}
\end{definition}

Let us now define two solution sets, one for the key equation,
and the other for a modified key equation. 
\begin{definition}\label{def:solmod}
{\rm
Let 
$$
M:=\big\{u\in \eftwos[X]\big|u'\equiv S u\mod (X^{2t})\big\},
$$
and
$$
N:=\bigg\{(u,v)\in \eftwos[X]\times \eftwos[X]\bigg|u\equiv
\frac{\even{S}(X)}{1+X\odd{S}(X)}v\mod (X^t)\bigg\}.
$$
}
\end{definition}

The following proposition relates $M$ and $N$. We note that the proof
is inspired by the proof of the lemma on p.~25 of \cite{Berl66}.
\begin{proposition}\label{prop:isom}
$M$ is an $\eftwos[X]$-submodule of $\tlftwomx$ and $\mu(N)=M$. Hence
$\mu$ restricts to an isomorphism $\mu'\colon N\isom M$, as in the
following commutative diagram: 
$$
\begin{gathered}
\xymatrix{
\eftwos[X]\times\eftwos[X]\ar[rr]^(0.7){\sim}_(0.7){\mu} && \tlftwomx\\
\\
N \ar[uu]^{\cup}\ar@{.>}[rr]^(0.7){\sim}_(0.7){\mu'} && M\ar[uu]^{\cup}
}
\end{gathered}
$$
\end{proposition}

\begin{proof}
While the first assertion follows from the second, it is worthwhile
(and straightforward) to prove it directly: First, $M$ is clearly an
abelian subgroup of $\tlftwomx$, and so it remains to verify that for
all $f\in \eftwos[X]$ and all $u\in M$, $f\cdot_1 u\in M$. But
\begin{multline*}
\big(f(X)\cdot_1 u(X)\big)' = \big(f(X^2)u(X)\big)' =f(X^2)u'(X)
\equiv f(X^2)Su = S\cdot \big(f(X)\cdot_1 u(X)\big),
\end{multline*}
where characteristic $2$ was used in the second equation, and where
$\equiv$ stands for congruence modulo $(X^{2t})$. 

Let us now turn to the second assertion. For a polynomial $u\in \eftwos[X]$
we have $u'(X)=\odd{u}(X^2)$, and so
$u$ is in $M$ iff 
\begin{eqnarray}
\odd{u}(X^2) & \equiv &
\big(\even{S}(X^2)+X\odd{S}(X^2)\big)\big(\even{u}(X^2)+X\odd{u}(X^2)\big)
\nonumber\\
 & = &
\even{S}(X^2)\even{u}(X^2)+X^2\odd{S}(X^2)\odd{u}(X^2) +\nonumber\\
&& X\Big(\odd{S}(X^2)\even{u}(X^2)+\even{S}(X^2)\odd{u}(X^2)\Big).\label{eq:double} 
\end{eqnarray}
Now,
(\ref{eq:double}) is equivalent to the following pair of equations:
\begin{equation}\label{eq:even}
\odd{u}(X^2)\equiv
\even{S}(X^2)\even{u}(X^2)+X^2\odd{S}(X^2)\odd{u}(X^2) \mod (X^{2t}),
\end{equation}
and
\begin{equation}\label{eq:odd}
0\equiv \odd{S}(X^2)\even{u}(X^2)+\even{S}(X^2)\odd{u}(X^2)\mod
(X^{2t}). 
\end{equation}

Hence, $u\in M$ iff (\ref{eq:even}) and (\ref{eq:odd}) hold. To prove
the assertion, we will show that for all $u\in 
\eftwos[X]$, (i)
(\ref{eq:odd}) follows from (\ref{eq:even}) (so that $u\in M$ iff
(\ref{eq:even}) holds), and (ii) (\ref{eq:even}) is equivalent to
$\mu^{-1}(u)\in N$ (so that $u\in M$ iff $\mu^{-1}(u)\in N$). 

Let us start with  (ii). Observe that for all $u$, (\ref{eq:even}) is
equivalent to 
\begin{equation}\label{eq:form}
\odd{u}(X)\equiv
\even{S}(X)\even{u}(X)+X\odd{S}(X)\odd{u}(X) \mod (X^{t}),
\end{equation}
Collecting terms in the last equation, and noting that
$1+X\odd{S}(X)$ indeed has an inverse in 
$\eftwos[X]/(X^t)$, we see that for all $u\in
\eftwos[X]$, (\ref{eq:even}) is equivalent to $\mu^{-1}(u)\in N$, as
required. 

Moving to (i), recall that in the binary case in question, for all 
$i$, $S_{2i}=S_i^2$, and hence 
\begin{eqnarray}
S(X)^2 & =   & S_2+S_4 X^2+\cdots+S_{2d-2}X^{2d-4} \nonumber \\
       & \equiv & \odd{S}(X^2)\mod (X^{2t}) \label{eq:sodd}
\end{eqnarray}
(recall that $d=2t+1$). Hence 
$$
\Big(\even{S}(X^2)+X\odd{S}(X^2)\Big)^2\equiv \odd{S}(X^2)\mod (X^{2t}),
$$
that is, 
$$
\Big(\even{S}(X^2)\Big)^2 \equiv X^2\Big(\odd{S}(X^2)\Big)^2 + 
\odd{S}(X^2)\mod (X^{2t}),  
$$
which is equivalent to 
\begin{equation}\label{eq:sqsynd}
\Big(\even{S}(X)\Big)^2\equiv X\Big(\odd{S}(X)\Big)^2 + \odd{S}(X)\mod
(X^{t}). 
\end{equation}
Now, to prove (i), assume that (\ref{eq:even}) holds, so that
(\ref{eq:form}) also holds. 
Multiplying (\ref{eq:form}) by $\even{S}(X)$ and writing ``$\equiv$''
for congruence modulo $(X^{t})$, we get
\begin{eqnarray*}
\odd{u}(X)\even{S}(X) & \equiv &
\Big(\even{S}(X)\Big)^2\even{u}(X)+X\odd{S}(X)\even{S}(X)\odd{u}(X) \\
& \stackrel{(*)}{\equiv} &\Big(1+X\odd{S}(X)\Big)
\odd{S}(X)\even{u}(X)+X\odd{S}(X)\even{S}(X)\odd{u}(X)
\end{eqnarray*}
(where $(*)$ follows from (\ref{eq:sqsynd})). Collecting terms and
canceling $1+X\odd{S}(X)$ (which is invertible in
$\eftwos[X]/(X^t)$), we get 
$$
0\equiv \odd{S}(X)\even{u}(X)+\even{S}(X)\odd{u}(X)\mod
(X^{t}),
$$
which is equivalent to (\ref{eq:odd}), as required.
\end{proof}

The following proposition shows that the isomorphism $\mu$ behaves
properly with respect to the order (equivalently, with respect to
$\wdegpm$) of a pair of polynomials. 

\begin{proposition}\label{prop:ordmu}
For $f\in \eftwos[X]$, it holds that 
\begin{equation}\label{eq:ordmu}
\ord\big(\mu^{-1}(f)\big) =
\left\lceil\frac{\deg(f)}{2}\right\rceil =
\left\lceil\frac{\ord(f',f)}{2}\right\rceil . 
\end{equation}
\end{proposition}

\begin{proof}
We have
\begin{eqnarray*}
\ord\big(\mu^{-1}(f)\big) & = & \ord\big(\odd{f},\even{f}\big)\\
& \stackrel{(\text{a})}{=} & \begin{cases}
\deg\big(\even{f}\big)& \text{if $2|\deg(f)$}\\
\deg\big(\odd{f}\big)+1&\text{if $2\nmid \deg(f)$}
\end{cases}\\
& = & \left\lceil\frac{\deg(f)}{2}\right\rceil, 
\end{eqnarray*}
where for (a), note that if $2|\deg(f)$, then
$\deg\big(\even{f}\big)>\deg\big(\odd{f}\big)$, while if $2\nmid
\deg(f)$, then $\deg\big(\odd{f}\big)\geq
\deg\big(\even{f}\big)$. 
\end{proof}

Let us now describe an application of Proposition \ref{prop:ordmu} for
HD unique decoding when the number $\veps$ of errors satisfies
$\veps\leq t$. Recall that in this case, up to a non-zero
multiplicative constant, $\sigma$ is the unique element
of $M$ with minimum degree \cite[Prop.~6.6]{Roth}.\footnote{
This Proposition from \cite{Roth} states that when $\veps\leq t$,
$(\omega,\sigma)$ is the unique order minimizing pair in
$M':=\{(u,v)|u\equiv Sv\mod(X^{2t})\}$, up to a non-zero
multiplicative constant. Hence in the binary case, where $w=\sigma'$,
$(\sigma',\sigma)$ is the unique order minimizing pair in
$\{(u',u)|u'\equiv Su\mod(X^{2t})\}\subseteq M'$.  
}

\begin{proposition}\label{prop:minord}
If $\veps\leq t$, then $\mu^{-1}(\sigma)$ has the minimum leading
monomial in $N\smallsetminus \{(0,0)\}$ with
respect to $<_{-1}$ (this identifies $\mu^{-1}(\sigma)$ uniquely
up to a non-zero multiplicative constant).
\end{proposition}

\begin{proof}
First, it follows from Proposition
\ref{prop:ordmu} that there is no element of $N\smallsetminus
\{(0,0)\}$ with a strictly lower order than
$\ord(\mu^{-1}(\sigma))$.   

If $\veps=\deg(\sigma)$ is even, then it follows from Proposition
\ref{prop:ordmu} that 
elements of $N$ with the same order as 
$\mu^{-1}(\sigma)$ must be of the form $\mu^{-1}(f)$ for $f\in M$ with
$\deg(f)\in\{\deg(\sigma),\deg(\sigma)-1\}$. By minimality, such an
$f$ must be $c\cdot \sigma$ for some $c\neq 0$.  
We conclude that when $\veps$ is even, $\mu^{-1}(\sigma)$ is
the unique (up to a multiplicative constant) 
order-minimizing element of $N$. 

On the other hand, if $\veps$ is odd, then it follows from
Proposition \ref{prop:ordmu} that the only elements of $N$ with the same
order as $\mu^{-1}(\sigma)$ (other than $\mu^{-1}(c\cdot \sigma)$)
must be of the form $\mu^{-1}(f)$ for 
$f\in M$ with $\deg(f)=\deg(\sigma)+1$. Such an $f$ has an even
degree, which implies that $\lmmo(\mu^{-1}(f))$ is on the
right. However, $\lmmo(\mu^{-1}(\sigma))$ is on the left (as
$\deg(\sigma)$ is odd), and also
$\lmmo(\mu^{-1}(\sigma)) = (X^{(\deg(\sigma)-1)/2},0) <_{-1}
(0,X^{(\deg(\sigma)+1)/2})=\lmmo(\mu^{-1}(f))$, 
by the definition of  
$<_{-1}$. 
\end{proof}

It follows from Proposition \ref{prop:minord} that HD decoding for up to
$t$ errors can be performed by finding a \grobner{} basis for the module $N$
with respect to $<_{-1}$, and taking $\sigma$ as the $\mu$-image of
the minimal element in the \grobner{} basis. 
Finding a \grobner{} basis for $N$ can be done, e.g., with any one of
the algorithms described in \cite{Fitz95}. Note that these algorithms
include a BM-like algorithm (Alg.~4.7), 
as well as an Euclidean algorithm (Alg.~3.7). 

As the problem dimensions are halved in $N$, this part
of HD decoding is substantially faster than decoding without
considering the binary alphabet: for a
typical $O(t^2)$ algorithm, the complexity is reduced by a
factor of about $4$. However, for performing the decoding 
in $N$, we need to calculate the {\bf modified syndrome} $\hat{S}(X)$,
defined as the unique polynomial of degree $<t$ in the image of
$$
R(X):=\frac{\even{S}(X)}{1+X\odd{S}(X)}\in\eftwos[[X]]
$$
in $\eftwos[X]/(X^t)$ (so that $\hat{S}(X)$ consists of the first $t$
terms of $R$). 

Writing 
\begin{eqnarray*}
R(X) & = & a_0+a_1X+a_2X^2+\cdots \\
\even{S}(X) &= & b_0 + b_1X+\cdots +b_{t-1}X^{t-1}\\
\odd{S}(X) & = & c_0 + c_1X+\cdots +c_{t-1}X^{t-1}
\end{eqnarray*}
(with $b_i=S_{2i+1}$ and $c_i=S_{2(i+1)}$ for all $i$), it follows
from  
$$
(a_0+a_1X+a_2X^2+\cdots)(1+c_0X+c_1X^2+\cdots+c_{t-1}X^t) = b_0 +
b_1X+\cdots +b_{t-1}X^{t-1}  
$$
that 
$a_0=b_0$, and for all $i\in\{1,\ldots,t-1\}$, $a_i$ can be recovered
recursively as
\begin{equation}\label{eq:recur}
a_i=b_i+a_{i-1}c_0+a_{i-2}c_1+\cdots+a_0 c_{i-1}.
\end{equation}

The overall number of multiplications required for calculation
$a_1,\ldots, a_{t-1}$ is therefore $t(t-1)/2$. Using the complexity
estimations from \cite{FJ98}, this is again about $1/4$ of the
complexity of decoding with \cite[Alg.~4.7]{Fitz95} without using the
binary alphabet. So, calculating the 
modified syndrome and a \grobner{}-basis for $N$
requires about half the operations for finding the ELP without taking
advantage of the binary alphabet. 

Similarly to \cite{JK95}, additional calculations are required if only
one coordinate is maintained in the \grobner{}-basis algorithm, and
its output  includes only the even part of the 
ELP; this is the case for \cite[Alg.~4.7]{Fitz95}, 
but not for the Euclidean algorithm \cite[Alg.~3.7]{Fitz95}. 
In such a case, the modified key equation should be used to calculate
the odd part of the ELP, at a cost of  $t(t+1)/2$ additional
multiplications. Hence in such a case, the gain over
decoding without using the binary alphabet is by factor of about 
$3/4$. 

We note that the idea of calculating a
modified syndrome and working with half the dimensions appears also in
\cite{Ca88} (with the Euclidean algorithm) and in \cite{JK95} (with
the BM algorithm), where the modified key equation appears implicitly,
including the recursion (\ref{eq:recur}). However, \cite{Ca88} and
\cite{JK95} do not show that this is a consequence of an isomorphism
between two solution modules, and that any algorithm for finding a 
\grobner{} basis can be used. Also, the derivation in
\cite{JK95} uses an entirely different method (building on Newton's
identities), and results in two different modified key equations for
even/odd $t$.  

We also note that Berlekamp's well-known 
method of using his algorithm only for half of the iterations in the
binary case \cite[pp.~24ff]{Berl66} is not applicable to other decoding
algorithms, such as the Euclidean algorithm.  Moving to the modified
syndrome enables to benefit from the binary alphabet with any
algorithm for computing a \grobner{} basis to the solution module of a
key equation.  

\section{A \grobner{}-bases formulation of the SD Wu list decoding
algorithm for binary BCH codes}\label{sec:softwu}
The \grobner{}-bases formulation of the (HD and SD) Wu list decoding
algorithm for binary BCH codes is missing in the literature. While
quite similar to the case of irreducible Goppa codes \cite{BHNW13},
there are some differences that must be considered, even for the HD
case. For example, while any non-zero polynomial has a multiplicative
inverse modulo an irreducible 
polynomial, this is not the case when working modulo $X^{d-1}$. As
another example, for BCH codes it seems more natural to use a
different monomial ordering than that of \cite{BHNW13}. Also,
considering SD decoding, where different coordinates may have
different multiplicities, a refinement of the interpolation theorem
is required. Hence,
for completeness we will briefly sketch the \grobner{} bases
formulation of the (SD) Wu list decoding algorithm for binary BCH
codes. 

The following proposition is a small adjustment of
\cite[Prop.~2]{BHNW13} to the monomial ordering of the current
paper. The proof is similar to that of \cite[Prop.~2]{BHNW13}, and is
omitted.   

\begin{proposition}{\cite[Prop.~2]{BHNW13}}\label{prop:degbounds}
Let $\{\bh_1=(h_{10},h_{11}),\bh_2=(h_{20},h_{21})\}$ be a \grobner{}
basis for the module $N$ of Definition \ref{def:solmod} with respect to
the monomial ordering $<_{-1}$, and assume w.l.o.g.~that $\lm(\bh_1)$
contains $(1,0)$ and $\lm(\bh_2)$ contains $(0,1)$. Then  

\begin{enumerate}

\item $\deg(h_{10})+\deg(h_{21})=t$.

\item For $\bs{u}=(u_0,u_1)\in N$, let $f_1,f_2\in \eftwos[X]$ be
the unique polynomials such
that $\bs{u}=f_1\bh_1+f_2\bh_2$. Put
$\veps':=\deg(\mu(\bs{u}))=\deg\big(u_1(X^2)+Xu_0(X^2)\big)$. 

\begin{enumerate}

\item If $\veps'$ is even, so that $\veps'=2\deg(u_1)$, then
\begin{eqnarray*}
\deg(f_1(X)) &\leq &\deg(u_1(X))-t+\deg(h_{21}(X))-1=:\woneeven\\
\deg(f_2(X)) &= &\deg(u_1(X))-\deg(h_{21}(X))=:\wtwoeven
\end{eqnarray*}

\item Otherwise, if $\veps'$ is odd, so that $\veps'=2\deg(u_0)+1$,
then 
\begin{eqnarray*}
\deg(f_1(X)) & = &\deg(u_0(X))-t+\deg(h_{21}(X))=:\woneodd\\
\deg(f_2(X)) &\leq &\deg(u_0(X))-\deg(h_{21}(X))=:\wtwoodd
\end{eqnarray*}

\end{enumerate}

\end{enumerate}

\end{proposition}

\begin{definition}\label{def:fiwi}
{\rm 

\begin{enumerate}

\item Fixing $\bh_i$ ($i=1,2$) as in Proposition
\ref{prop:degbounds}, let $\hath_i:=\mu(\bh_i)$, so that $\hath_i$
is the polynomial obtained by gluing the odd and even parts specified
in $\bh_i$.

\item Take $\bs{u}:=\mu^{-1}(\sigma)$ in part 2 of Proposition
\ref{prop:degbounds}, so that $\veps'=\deg(\sigma)=\veps$ is the
number of errors in the received word $\bs{y}$. For $i=1,2$, let
$f_i(X)\in\eftwos[X]$ be the unique polynomials from the proposition
for this choice of $\bs{u}$, and let 
$$
w_i:=\begin{cases}
\wieven{i} &\text{if }\veps\text{ is even}\\
\wiodd{i}  &\text{if }\veps\text{ is odd}.
\end{cases}
$$ 

\end{enumerate}
}
\end{definition}

\begin{remark}\label{rem:wsum}
{\rm 
Note that $w_1+w_2=\veps-t-1$, regardless of the parity of $\veps$.
}
\end{remark}

To continue, we note that since $\mu$ restricts to an isomorphism
$N\isom M$, $M$ is free of rank $2$, and $\mu$ induces a bijection
between the set of bases of $N$ and the set of bases of $M$. The
following proposition is a counterpart to the last part of the proof
of \cite[Prop.~5.27]{Niel}. 

\begin{proposition}\label{prop:gcd}
If $\{g_1(X),g_2(X)\}$ is any free-module basis for $M$ (as a
submodule of $L$), then $\gcd(g_1,g_2)=1$.
\end{proposition}

\begin{proof}
Note that the $L$-submodule generated by a set of polynomials is
contained in the ideal generated by the same polynomials. Hence if
$r(X)$ is a common factor of $g_1,g_2$, then $r(X)$ is a factor of
every element of $M$. To prove the assertion, it is therefore
sufficient to show that there are \emph{some} coprime polynomials in
$M$. Clearly, $X^{2t+1}\in M$. Also, using (\ref{eq:sodd}), it can be
verified that $1+XS(X)\in M$. 
\end{proof}

Recall from Definition \ref{def:fiwi}
that $f_1,f_2\in \eftwos[X]$ are the unique polynomials such
that $\mu^{-1}(\sigma)=f_1\bh_1+f_2\bh_2$, and that
$\hath_i:=\mu(\bh_i)$, $i=1,2$. Since $\mu$ is a
homomorphism, we have
\begin{multline*}
\sigma=\mu(f_1\bh_1+f_2\bh_2) = f_1(X)\cdot_1\mu(\bh_1) +
f_2(X)\cdot_1\mu(\bh_2) 
= f_1(X^2)\hath_1+f_2(X^2)\hath_2,
\end{multline*}
that is, 
\begin{equation}\label{eq:sig}
\sigma(X)=f_1(X^2)\hath_1(X)+f_2(X^2)\hath_2(X).
\end{equation}
Hence, 
\begin{equation}\label{eq:preinterp}
\forall x\in\roots(\sigma,\eftwos),\quad
f_1(x^2)\cdot\hath_1(x) +
f_2(x^2)\cdot\hath_2(x)=0,
\end{equation}
and also, by Proposition \ref{prop:gcd},
\begin{equation}\label{eq:preinterpt}
\forall x\in\eftwos^*,\quad \hath_1(x)=0 \implies \hath_2(x)\neq 0.
\end{equation}

\begin{remark}\label{rem:coprime}
{\rm
Since every element of $\eftwos$ has a square root, we can write 
$f_i(X^2)=g_i(X)^2$ for some $g_i$, $i=1,2$. As noted in
\cite{BHNW13}, since $\sigma$ is square-free, we must therefore have 
$\gcd(f_1,f_2)=1$. 
}
\end{remark}

\begin{remark}\label{rem:root}
{\rm
It follows from (\ref{eq:preinterpt}) that for all
$\bs{g}\in N$, written uniquely as
$\bs{g}=g_0\bs{h}_1+g_1\bs{h}_2$,\footnote{We use subscripts starting
from $0$ to comply with the notation of the following sections, where
the subscript stands for the power of the intermediate $Y$.} 
it holds that for all $x\in \eftwos$,  
\begin{equation}\label{eq:}
\mu(\bs{g})(x) = 0 \quad\iff\quad 
\begin{cases}
g_0(x^2)+\frac{\hath_2(x)}{\hath_1(x)}g_1(x^2)=0 & \text{if
}\hath_1(x)\neq 0\\
g_1(x^2) = 0 & \text{if }\hath_1(x)= 0.
\end{cases}
\end{equation}
}
\end{remark}

Note that $\bh_1$ and $\bh_2$ depend only on the received word
(through the syndrome), and the decoder's task is to find $f_1,f_2$,
as $\sigma$ can be restored from $f_1,f_2$. The interpolation theorem
indicates how (\ref{eq:preinterp}) and (\ref{eq:preinterpt}) can be
used for finding $f_1$ and $f_2$.

The following interpolation theorem is a refinement of a theorem of
Trifonov \cite{Trif10}. We note that although for decoding we are only
interested in the case where $u_1,u_2$ in the theorem are coprime,
this assumption is not required for the theorem. For completeness, we
include here the simple proof.

\begin{theorem}{\cite{Trif10}}\label{thm:interp} 
For any field $K$, let $Q(X,Y,Z)\in \homog^K$, and let $u_1(X),u_2(X)\in K[X]$ be
polynomials with
$\deg(u_1(X))\leq w'_1$ and $\deg(u_2(X))\leq w'_2$ for some $w'_1,w'_2\in
\bbQ$. For $e\in\bbNp$, let $\{(x_i,y_i,z_i)\}_{i=1}^{e}\subseteq K^3$
be a set of triples satisfying the following conditions:
\begin{enumerate}

\item The $x_i$ are distinct.

\item The $y_i$ and $z_i$ are not simultaneously zero, that
is, for all $i$, $y_i=0\implies z_i\neq 0$.

\item For all $i$, $z_iu_1(x_i)+y_iu_2(x_i)=0$.

\end{enumerate}
Then, if 
$$
\sum_{i=1}^e\mult\big(Q(X,Y,Z),(x_i,y_i,-z_i)\big) >
\wdeg_{1,w'_1,w'_2}\big(Q(X,Y,Z)\big), 
$$
then $Q(X,u_1(X),u_2(X))=0$ in $K[X]$.
\end{theorem}

\begin{proof}
Write $I:=\{1,\ldots,e\}$, $I_1:=\{i\in I, z_i\neq 0\}$ and $I_2:=I\smallsetminus
I_1=\{i\in I, z_i=0\}$.  Note that by assumption, for all $i\in I_2$,
$y_i\neq 0$. Then 
\begin{equation}\label{cw_eq:forall}
\forall i\in I_1,
u_1(x_i)=-\frac{y_iu_2(x_i)}{z_i}\qquad\text{and}\qquad\forall i\in I_2,
u_2(x_i)=-\frac{z_i u_1(x_i)}{y_i}.
\end{equation}
Now, it follows from Proposition \ref{cw_prop:compmult} that
\begin{equation}\label{cw_eq:firststep}
\sum_{i=1}^e \mult\big(Q(X,u_1(X),u_2(X)), x_i\big) \geq  \sum_{i=1}^e
\mult\big(Q(X,Y,Z), (x_i,u_1(x_i),u_2(x_i))\big).
\end{equation}
Writing $m$ for the right-hand side of (\ref{cw_eq:firststep}), we have
\begin{eqnarray*}
m&\stackrel{(\text{a})}{=}&  \sum_{i\in
I_1}\mult\Big(Q,\big(x_i,-\frac{y_iu_2(x_i)}{z_i},u_2(x_i)\big)\Big)
+  \sum_{i\in I_2}\mult\Big(Q,\big(x_i,u_1(x_i),-\frac{z_i
u_1(x_i)}{y_i}\big)\Big) \\
& \stackrel{(\text{b})}{=} & \sum_{i\in
I_1}\mult\big(Q,(x_i,-y_iu_2(x_i),z_iu_2(x_i))\big)
+ 
\sum_{i\in I_2}\mult\big(Q,(x_i,y_iu_1(x_i),-z_i
u_1(x_i))\big) \\
& \stackrel{(\text{c})}{\geq} & \sum_{i\in
I_1}\mult\big(Q,(x_i,y_i,-z_i)\big) 
+ 
\sum_{i\in I_2}\mult\big(Q,(x_i,y_i,-z_i)\big)\\
& \stackrel{(\text{d})}{>} & \wdeg_{1,w'_1,w'_2}\big(Q(X,Y,Z)\big)
\stackrel{(\text{e})}{\geq} \deg\big(Q(X,u_1(X),u_2(X)\big), 
\end{eqnarray*}
where (a) follows from (\ref{cw_eq:forall}), (b) follows from part 1 of
Proposition \ref{cw_prop:yzhomog}, (c) follows from parts 1 and 4 of
Proposition \ref{cw_prop:yzhomog} (part 4 is required if $u_2(x_i)=0$
or $u_1(x_i)=0$ for some $i$), (d) is by assumption, and (e) holds
since $\deg f_i\leq w'_i, i=1,2$. It follows that $Q(X,u_1(X),u_2(X))$
is the 
zero polynomial. 
\end{proof}

\begin{corollary}\label{coro:interp}
For $Q(X,Y,Z)\in \homog^{\eftwos}$ and for the $f_i,w_i,\hath_i$
($i=1,2$) from Definition \ref{def:fiwi}, suppose that 
$$
\sum_{x\in\roots(\sigma)}\mult\Big(Q(X,Y,Z),
\big(x^2,\hath_2(x),\hath_1(x)\big)\Big) >
\wdeg_{1,w_1,w_2}\big(Q(X,Y,Z)\big). 
$$
Then $Q(X,f_1(X),f_2(X))=0$. Consequently, if $f_2(X)\neq 0$, then 
$\big(Yf_2(X)+f_1(X)\big)|\downed{Q}(X,Y)$ in $\eftwos[X,Y]$. 
\end{corollary}

\begin{proof}
The first assertion is a straightforward consequence of
(\ref{eq:preinterp}), (\ref{eq:preinterpt}), and Theorem
\ref{thm:interp}. For the second assertion, let $\rho$ be the
$(Y,Z)$-homogeneous degree of $Q$. Then 
$Q(X,Y,Z)=Z^{\rho}\downed{Q}(X,Y/Z)$, and hence if $f_2(X)\neq 0$,
$\downed{Q}(X,f_1(X)/f_2(X))=0$. It follows that $Y+f_1(X)/f_2(X)$
divides $\downed{Q}(X,Y)$ in $\eftwos(X)[Y]$.  Hence, we can write
$$
\downed{Q}(X,Y) = \frac{u(X,Y)}{v(X)}\big(Y+\frac{f_1(X)}{f_2(X)}\big)
 =  \frac{u_1(X,Y)}{v_1(X)}(Yf_2(X)+f_1(X)),
$$
for some $u,u_1\in\eftwos[X,Y], v,v_1\in\eftwos[X]$, and where
$u_1(X,Y)/v_1(X)$ is a reduced fraction. As $f_1,f_2$ are coprime and
$\downed{Q}(X,Y)\in\eftwos[X,Y]$, $v_1(X)$ must be a constant.   
\end{proof}

\begin{definition}
{\rm
For $\bm=\{m_x\}_{x\in\eftwos^*}\in \bbN^{\eftwos^*}$, let 
$$
\am:=\Big\{Q\in \eftwos[X,Y,Z]\Big|\forall x\in
\eftwos^*,\mult\Big(Q,\big(x^2,\hath_2(x),\hath_1(x)\big)\Big)\geq
m_{x^{-1}}\Big\}
$$
be the ideal of all polynomials having the multiplicities specified in
$\bm$. Also,
For $\rho\in\bbN$, let $M(\bm,\rho)$ be the $\eftwos[X]$-submodule of
$\eftwos[X,Y,Z]$ defined by 
$$
M(\bm,\rho):=\am\cap \homog^{\eftwos}(\rho).
$$
}
\end{definition}

To continue, it will be convenient to introduce some notation similar
to that of \cite{KV03}. Regarding the following definition, we note
that it is important to take special care when considering the number
of free variables with non-integer weights. In this context,
we will restrict attention to weights (and weighted degrees) in
$\frac{1}{2}\bbN$. 

\begin{definition}\label{def:nvar}
{\rm
\begin{enumerate}
\item For $\rho\in\bbN$ and $w'_1,w'_2,d\in \frac{1}{2}\bbN$, let 
$$
\nvarsrwt(d):=\big|\big\{(i,j)\in\bbN^2\big|i+w'_1j+w'_2(\rho-j)\leq
d\text{ and }j\leq \rho\big\}\big| 
$$
be the number of free coefficients in a polynomial $Q\in
\homog(\rho)$ with $\wdeg_{1,w'_1,w'_2}(Q)\leq d$. Note that
$\nvarsrwt(d)$ is always an integer.

\item For $\eneq\in \bbN$, let 
$\mindegrwt(\eneq):=\min\big\{d\in
\frac{1}{2}\bbN\big|\nvarsrwt(d)>\eneq\big\}$ 
be the minimum $(1,w'_1,w'_2)$-weighted degree required for getting at
least $\eneq+1$ free coefficients in a polynomial $Q\in\homog(\rho)$. 

\item For $\bm=\{m_x\}_{x\in\eftwos^*}$, let $\cost(\bm) :=
\frac{1}{2}\sum_{x\in\eftwos^*}m_x(m_x + 1).$\footnote{Recall that
$\cost(\bm)$ is the number of  
homogeneous linear equations on the coefficients of a bivariate 
polynomial $P\in \eftwos[X,Y]$ describing the condition
$\forall x\in \eftwos^*, \mult(P,(x,y_x))\geq m_x$
(for some family $\{y_x\}_x$).}

\end{enumerate}
}
\end{definition}

The following simple bound on $\mindegrwt(\eneq)$ appears implicitly
in \cite{Trif10} for integer weights. It also appears implicitly in
\cite{BHNW13}, \cite{Niel}
for non-integer weights.  

\begin{proposition}\label{prop:delta}  
With the notation of Definition \ref{def:nvar}, suppose that exactly
one of $w'_1,w'_2$ is not an integer (and hence equals an odd multiple
of $1/2$). Then 
\begin{equation}\label{eq:half}
\mindegrwt(\eneq)\leq \begin{cases}
\displaystyle
\frac{\eneq}{\rho+1}+\frac{\rho}{2}(w'_1+w'_2) & \text{if $\rho$ is
odd}\\ \displaystyle
\frac{\eneq+1/2}{\rho+1}+\frac{\rho}{2}(w'_1+w'_2) & \text{if $\rho$ is
even}.
\end{cases}
\end{equation}
If both $w'_1,w_2$ are integers, then the first case of
(\ref{eq:half}) holds, with an integer $\mindegrwt(\eneq)$.  
\end{proposition}
Since the case of non-integer weights is subtle, we include the proof
of Proposition \ref{prop:delta} in Appendix \ref{app:nvar}. 

\begin{corollary}\label{coro:cond}
Let $\bm\in \bbN^{\eftwos^*}$  be a vector of multiplicities, let $\rho\in
\bbN$, and let $f_i,w_i,\hath_i$ ($i=1,2$) be as in  
Definition \ref{def:fiwi}. Then,
\begin{enumerate}

\item There exists a non-zero polynomial $P\in M(\bm,\rho)$ with
$\wdeg_{1,w_1,w_2}(P)\leq  \mindegrw(\cost(\bm))$.

\item Let $P\in M(\bm,\rho)$ be a polynomial satisfying the condition
in part 1. Then, if
\begin{equation}\label{eq:mcond}
\sum_{x\in \roots(\sigma)}m_{x^{-1}} > \mindegrw(\cost(\bm)),
\end{equation}
then $P(X,f_1(X),f_2(X))=0$. In particular, if (\ref{eq:mcond}) holds,
then a non-zero polynomial $Q\in M(\bm,\rho)$ of minimum
$(1,w_1,w_2)$-weighted degree satisfies $Q(X,f_1(X),f_2(X))=0$. 

\end{enumerate}
\end{corollary}

\begin{proof}
1. Follows from the usual ``more variables than equations'' argument
(recalling Proposition \ref{cw_prop:yzhomog}). 

2. Writing $\bs{p}_x:=\Big(x^2,\hath_2(x),
\hath_1(x)\Big)$, we have 
$$
\sum_{x\in \roots(\sigma)}\mult(P,\bs{p}_x)  \geq 
\sum_{x\in \roots(\sigma)}m_{x^{-1}}  
> \mindegrw(\cost(\bm))\geq \wdeg_{1,w_1,w_2}\big(P\big).
$$
The proposition now follows from Corollary \ref{coro:interp}.   
\end{proof}

Corollary \ref{coro:cond} indicates that the decoder should search for
a non-zero $Q\in M(\bm,\rho)$ of minimum $(1,w_1,w_2)$-weighted
degree.  Up until now, we considered partially homogeneous trivariate
polynomials (following Trifonov) for avoiding points at infinity.
While the most efficient algorithms for list decoding 
require this setup, for the fast Chase decoding
algorithm of the following section it is sufficient to consider
\kotter{}'s iteration. For this purpose, it will be convenient to
work with bivariate polynomials, by de-homogenizing. 

\begin{proposition}\label{prop:tlm}
Let $\downed{M}(\bm,\rho)$ be the image of $M(\bm,\rho)$ under the
isomorphism $\downed{(\cdot)}$ of Definition 
\ref{cw_def:tlf}. Then $\downed{M}(\bm,\rho)$ is the
$\eftwos[X]$-module of all polynomials $P\in \eftwos[X,Y]$ satisfying
\begin{enumerate}

\item $\wdeg_{(0,1)}(P)\leq \rho$,

\item for all $x$ with $\hath_1(x)\neq 0$, 
$$
\mult\Bigg(P,\bigg(x^2,\frac{\hath_2(x)}{\hath_1(x)}\bigg)\Bigg)\geq
m_{x^{-1}}, 
$$

\item for all $x$ with $\hath_1(x)= 0$ (and hence with
$\hath_2(x)\neq 0$), 

$$
\mult\Big(Y^{\rho}P(X,Y^{-1}),(x^2,0)\Big)\geq m_{x^{-1}}. 
$$
\end{enumerate}
Moreover, for a polynomial $R\in M(\bm,\rho)$, we have
$$
\wdeg_{1,w_1,w_2}(R)=\wdeg_{1,w_1-w_2}(\downed{R})+w_2\rho. 
$$
Hence, if $P\in \downed{M}(\bm,\rho)\smallsetminus\{0\}$ has the
minimal $(1,w_1-w_2)$-weighted degree in
$\downed{M}(\bm,\rho)\smallsetminus\{0\}$, then 
the $R\in M(\bm,\rho)$ with $\downed{R}=P$ has the minimal
$(1,w_1,w_2)$-weighted degree in $M(\bm,\rho)\smallsetminus\{0\}$. 
\end{proposition}

\begin{proof}
The characterization of $\downed{M}(\bm,\rho)$ Follows from
Proposition \ref{cw_prop:yzhomog}, and the second assertion is easily
verified. 
\end{proof}

Note that
\begin{equation}\label{eq:deltaw}
w_1-w_2=\begin{cases}
2\deg(h_{21})-t-1 & \text{if } \veps\text{ is even}\\
2\deg(h_{21})-t & \text{if } \veps\text{ is odd}.
\end{cases}
\end{equation}
Hence $w_1-w_2$ depends only on the parity of $\veps$, and the decoder
may know this parity if the even subcode of the BCH code considered
up to this point is used instead of the BCH code itself (at the cost
of loosing a single information bit). Alternatively, one may use a
non-integer weight, as suggested in \cite{BHNW13}, in order to avoid
loosing an information bit. We note that while for list decoding
there is a small gain in knowing the parity of $\veps$, this is not
the case for the fast Chase algorithm of the following section. Hence,
we do \emph{not} require working with the even subcode.

Let us now consider multiplicity vectors that take only a single
non-zero value, say, $m$. Typically and informally, this non-zero
multiplicity is assigned to all coordinates whose reliability is below
some threshold.   

\begin{corollary}\label{coro:simpcond}
Suppose that $I\subseteq \eftwos^*$ is the set of erroneous
coordinates (so that $|I|=\veps$). Let  $J\subseteq \eftwos^*$ be
some subset, and for $m\in \bbNp$, define $\bm$ by setting, for all
$x\in \eftwos^*$, 
$$
m_x:=\begin{cases}
m & \text{if } x \in J,\\
0 & \text{otherwise}.
\end{cases}
$$
Let also $\rho\in \bbNp$, and put $w:=2\deg(h_{21})-t-1/2$. Then if  
\begin{equation}\label{eq:listcond}
m\cdot|I\cap J|>\begin{cases}
\displaystyle
\frac{|J|\cdot
m(m+1)}{2(\rho+1)}+\frac{\rho}{2}(\veps-t-1/2) & \text{if }\rho \text{
is odd}\\
\displaystyle
\frac{|J|\cdot
m(m+1)+1}{2(\rho+1)}+\frac{\rho}{2}(\veps-t-1/2) & \text{if }\rho
\text{ is even}
\end{cases}
\end{equation}
then a non-zero polynomial $Q(X,Y)$ minimizing the
$(1,w)$-weighted degree in $\downed{M}(\bs{m},\rho)$ satisfies
$\big(Yf_2(X)+f_1(X)\big)|Q(X,Y)$ when $f_2(X)\neq 0$.
\end{corollary}
\begin{proof}
First note that when $\veps$ is even, we may take $w_1+1/2$ instead of
$w_1$ as an upper bound on $\deg(f_1)$, while when $\veps$ is odd, we
may take $w_2+1/2$ instead of $w_2$ as an upper bound on
$\deg(f_2)$. In both cases, the result is adding $1/2$ to the sum
$w_1+w_2=\veps-t-1$ (Remark \ref{rem:wsum}), while replacing $w_1-w_2$
by $2\deg(h_{21})-t-1/2=:w$, regardless of the error parity. The
assertion now follows from from Corollary \ref{coro:cond}, Proposition
\ref{prop:tlm}, and Proposition \ref{prop:delta}.
\end{proof}

For the fast Chase decoding algorithm of the following section, we
take $\rho=1$, and hence only the case of odd $\rho$ is relevant. It
should be noted that the seemingly negligible difference between the
odd and even $\rho$ cases becomes significant for the fast Chase
algorithm, where $\rho=m=1$; with the condition for even $\rho$,
Theorem \ref{thm:success} ahead would not hold. 

\begin{remark}
{\rm
Corollary \ref{coro:simpcond} is valid also when $\veps\leq 
t$, in which case it holds also for $f_2=0$, as we shall now
explain. When $\veps\leq t$, (\ref{eq:listcond}) holds for
$J=\emptyset$, for which $\downed{M}(\bs{m},\rho)$ is just the set of
polynomials of $Y$-degree at most $\rho$ in $\eftwos[X,Y]$ . There are
two cases to consider, depending on the sign of $w$ (note that by
definition, $w$ cannot be $0$). If $w>0$, then the minimum
$(1,w)$-weighted degree is achieved 
exactly by any non-zero constant polynomial. Factoring, we find that
$Yf_2+f_1$ must also be a non-zero constant. In vector form,
$(f_1,f_2)= (c,0)$ for some $c\in \eftwos^*$. Applying $\mu$, we
obtain that in this case, $\sigma=c\cdot \hath_1$. On the other hand,
if $w<0$, then the minimum 
$(1,w)$-weighted degree is achieved exactly by the polynomials of the
form $cY^{\rho}$ ($c\in\eftwos^*$), and factoring gives
$Yf_2+f_1=c'Y$ for some $c'$. In vector form, $(f_1,f_2)=(0,c')$, so that
$\sigma=c'\cdot \hath_2$. Note that as
$w=\deg(h_{21})-\deg(h_{10})-1/2$, we have $w>0\iff
\lmmo{\bh_1}<_{-1}\lmmo{\bh_2}$, and the above agrees with Proposition
\ref{prop:minord}.
}
\end{remark}

To complete the description of the SD Wu list decoding algorithm, one
has to specify methods for translating the channel reliability
information into the multiplicity vector $\bs{m}$ in
Corollary \ref{coro:cond}, to find an appropriate value of
the list size $\rho$, to consider efficient interpolation
(i.e., finding the minimizing $Q$), efficient factorization,
etc.. However, this is outside the main scope of the current
paper. Also, most of the above questions can be answered using small 
modifications of existing works, such as \cite{KV03}, \cite{Wu08},
\cite{BHNW13} and others. We therefore omit further details.

\section{The fast Chase algorithm}\label{sec:chase}

\subsection{Fast Chase decoding on a tree}\label{subsec:tree}
In the original Chase-II decoding algorithm \cite{Chase}
for decoding a binary code of minimum distance $d$, all possible
$2^{\lfloor d/2 \rfloor}$ error patterns on the $\lfloor d/2 \rfloor$
least reliable coordinates are tested (i.e., subtracted from the received
word). Bounded distance decoding is performed for each tested error
pattern, resulting in a list of $\leq 2^{\lfloor d/2 \rfloor}$ 
candidate codewords. If the list is non-empty, then the most
likely codeword in the list is chosen as the decoder output. 

The type of Chase algorithm considered in
the current paper is the following variant of the Chase-II
algorithm. As in \cite{Chase}, it is assumed
that the decoder has probabilistic reliability information on the
received bits. Using this information, the decoder identifies the
$\eta$ least reliable coordinates  
for some fixed  $\eta\in \bbNp$. Labeling coordinates by elements of
$\eftwos^*$ and writing $\alpha_1,\ldots,\alpha_{\eta}$ for the least
reliable coordinates, let $U:=\{\alpha_1,\ldots,\alpha_{\eta}\}$.  
 
The Chase decoding considered in the current paper runs over all test
error patterns on $U$ that have a weight of up to some
$\rmax\in \{1,\ldots,\eta\}$, and performs (the equivalent of) bounded
distance decoding for each such pattern. When $\rmax=\eta$,
the examined error patterns are all the error patterns on $U$. 

As in \cite{Wu12}, \cite{SB21}, a directed tree  
$T$ of depth $\rmax$ is constructed as follows. The root is the
all-zero vector in $\eftwo^{\eta}$, and for all $r\in 
\{1,\ldots,\rmax\}$, the  vertices at depth $r$ are the vectors of
weight $r$ in $\eftwo^{\eta}$.  
To define the edges of $T$, for each $r\geq
1$ and for each vertex
$\bs{\beta}=(\beta_1,\ldots,\beta_{\eta})\in\eftwo^{\eta}$ at 
depth $r$ with non-zero entries at coordinates $i_1,\ldots,i_r$, we
pick a single vertex $\bs{\beta}'=(\beta'_1,\ldots,\beta'_{\eta})$ at
depth $r-1$ that is equal to $\bs{\beta}$ on all coordinates, except
for one $i_{\ell}$ ($\ell\in\{1,\ldots,r\}$), for which
$\beta'_{i_{\ell}}=0$. Note that there are $r$ distinct ways to choose
$\bs{\beta}'$ given $\bs{\beta}$, and we simply fix one such 
choice. Now the edges of $T$
are exactly all such pairs $(\bs{\beta}',\bs{\beta})$. For an example
of the tree $T$, we refer to \cite[Fig.~1]{Wu12},
\cite[Fig.~1]{SB21}. 

The edge $(\bs{\beta}',\bs{\beta})$ defined above corresponds
to adjoining exactly one additional flipped coordinate (i.e.,
$\alpha_{i_{\ell}}$). Hence, the edge
$(\bs{\beta}',\bs{\beta})$ can be alternately marked by
$(\bs{\beta'},\alpha_{i_{\ell}})$. Similarly, we can identify a 
path from the root to a vertex at depth $r\geq 1$ (and hence the 
vertex itself) with a vector
$(\alpha_{i_{1}},\ldots,\alpha_{i_{r}})$
with distinct $\alpha_{i_{\ell}}$'s.

The main ingredient of \cite{Wu12} and \cite{SB21}, as well as
of the current paper, is an efficient algorithm for
updating polynomials for adding a single
modified coordinate $\alpha_{i_r}$. The tree $T$ is then traversed depth
first, saving intermediate results on vertices whose out degree is
larger than $1$, and applying the polynomial-update algorithm on the
edges. Because the tree is traversed depth first and has depth $\rmax$,
there is a need to save at most $\rmax$ vertex calculations at each
time (one for each depth).\footnote{This observation is due to
I.~Tamo.} See, e.g., \cite[Sec.~4.3]{SB21} for details.

\subsection{The update rule on an edge of the decoding
tree}\label{subsec:update} 
In this subsection we will use the results on SD Wu list decoding from
the previous section in order to define appropriate polynomials and an
update rule for a tree-based fast Chase decoding of binary BCH codes.

\begin{definition}
{\rm
For a subset $J\subseteq \eftwos^*$, let $L(J):=\downed{M}(\one_J,1)$,
where $\one_J$ is the vector $\{m_x\}_{x\in \efq^*}$ with
$$
m_x:=\begin{cases}
1 & \text{if } x \in J,\\
0 & \text{otherwise}.
\end{cases}
$$ 
For simplicity, we will sometimes write $L(x_1,\ldots,x_{\ell})$ for
$L(\{x_1,\ldots,x_{\ell}\})$, where $\ell\in \bbN$ and
$x_1,\ldots,x_{\ell}\in \efq^*$ are distinct.
} 
\end{definition}

Writing $J^{-1}_1:=\{x|x^{-1}\in J\text{ and }\hath_1(x)\neq 0\}$
and $J^{-1}_2:=\{x|x^{-1}\in J\text{ and }\hath_1(x)= 0\}$,
it follows from Proposition \ref{prop:tlm} that $L(J)$ consists of all
polynomials $g_0(X)+Yg_1(X)$ that satisfy the following two
conditions:
\begin{equation}\label{eq:deltaa}
\forall x\in J^{-1}_1,\quad  g_0(x^2)+\frac{\hath_2(x)}{\hath_1(x)}
g_1(x^2) = 0
\end{equation}
and
\begin{equation}\label{eq:deltab}
\forall x\in J^{-1}_2,\quad g_1(x^2)=0.
\end{equation}

\begin{remark}\label{rem:lj}
{\rm
Recalling Remark \ref{rem:root}, it follows from (\ref{eq:deltaa}) and
(\ref{eq:deltab}) that $L(J)$ consists of all polynomials
$g_0(X)+Yg_1(X)$ such that $\mu(g_0\bh_1+g_1\bh_2)$ vanishes on
$J^{-1}:=\{x|x^{-1}\in J\}=J^{-1}_1\cup J^{-1}_2$. 
}
\end{remark}

For the fast Chase decoding algorithm of the current paper, $J$ will
be the support of a test error pattern. Recall that
$f_1,f_2\in\eftwos[X]$ are such that
$\sigma(X)=f_1(X^2)\hath_1(X)+f_2(X^2)\hath_2(X)$. Recall also from
Corollary \ref{coro:simpcond} that $w:=2\deg(h_{21})-t-1/2$. The
following theorem shows that when the maximum list size is $1$, the SD
Wu list decoding algorithm becomes a means for flipping coordinates:
assigning a multiplicity of $1$ to a coordinate has the same effect as
flipping it in a Chase-decoding trial.
\begin{theorem}\label{thm:success}
Suppose that $f_2(X)\neq 0$. For $J\subseteq \eftwos^*$, let $n_1$ be
the number of erroneous coordinates on $J$, and let $n_2$ be the
number of correct coordinates on $J$ (so that $|J|=n_1+n_2$). Then if
$n_1-n_2\geq \veps -t$ and $\veps\geq t+1$, then for a non-zero polynomial
$g_0(X)+Yg_1(X)$ that minimizes the $(1,w)$-weighted degree in $L(J)$,
it holds that 
\begin{equation}\label{eq:propeq}
\big(f_1(X)+Yf_2(X)\big)\big|\big(g_0(X)+Yg_1(X)\big).
\end{equation}
Hence there exists some $t(X)\in\eftwos[X]$
such that 
$$
g_0(X)+Yg_1(X) = t(X)\big(f_1(X)+Yf_2(X)\big),
$$
so that 
\begin{equation}\label{eq:mug}
\mu(g_0\bh_1 +g_1\bh_2)=t(X^2)\sigma(X).
\end{equation}
Moreover, $\gcd\big(t(X^2),\sigma(X)\big)=1$ and $t(X^2)$ has no
roots outside $J^{-1}$. 
\end{theorem}

\begin{proof}
Suppose that $I\subseteq \eftwos^*$ is the set of erroneous
coordinates (so that $|I|=\veps$). Then by definition, $|I\cap
J|=n_1$, and (\ref{eq:listcond}) reads
$$
n_1>\frac{n_1+n_2}{2}+\frac{1}{2}(\veps-t-1/2),
$$
which is equivalent to $n_1-n_2>\veps-t-1/2$. Hence, this condition
holds by assumption, and the first assertion follows from Corollary
\ref{coro:simpcond}.  

Next, we would like to prove that
$\gcd\big(t(X^2),\sigma(X)\big)=1$. Suppose not. Then for some error
location $\beta$, it holds that $(X+\beta^{-1})|t(X^2)$. As before, we
may write $t(X^2)=\tl{t}(X)^2$ for some $\tl{t}(X)$, so that
$(X+\beta^{-1})^2 = X^2+\beta^{-2}$ divides $t(X^2)$. 
Hence, $X+\beta^{-2}$ divides $t(X)$. Moreover, 
\begin{equation}\label{eq:mubeta}
\mu\Big(\frac{g_0(X)}{X+\beta^{-2}}\bh_1 +
\frac{g_1(X)}{X+\beta^{-2}}\bh_2 \Big)
=\frac{t(X^2)}{(X+\beta^{-1})^2}\sigma(X)  
\end{equation}
has the same roots as $\mu(g_0\bh_1 +g_1\bh_2)$. It therefore follows 
from Remark \ref{rem:lj} that
\begin{equation}\label{eq:divided}
\frac{g_0(X)}{X+\beta^{-2}} + Y\frac{g_1(X)}{X+\beta^{-2}}\in L(J),
\end{equation}
contradicting the minimality of $g_0(X)+Yg_1(X)$.

Finally, suppose that there exists some $\beta\notin J$ such that
$\beta^{-1}$ is a root of $t(X^2)$. As above, it follows that
$(X+\beta^{-2})|t(X)$, and (\ref{eq:mubeta}) shows that the left-hand
side of (\ref{eq:divided}) is in $L(J)$, as $\mu(\cdots)$ in
(\ref{eq:mubeta})  has the same roots as $\mu(g_0\bh_1 +g_1\bh_2)$ on
$J^{-1}$. Again, this contradicts the minimality of $g_0(X)+Yg_1(X)$.
\end{proof}

In Theorem \ref{thm:success}, we have considered only the case
$f_2\neq 0$. If $f_2=0$, then it follows from Remark \ref{rem:coprime}
that $f_1$ must be some non-zero constant from $\eftwos^*$, so that
$\sigma(X)=c\cdot \hath_1(X)$ for some $c\in \efq^*$. In general
SD Wu list decoding, one has to check directly whether $\hath_1(X)$ is
a valid ELP by performing exhaustive
substitution, etc.~(even when the degree is beyond $t$), similarly to
Step 3 of 
\cite[Alg.~2]{BHNW13}. However, for the fast Chase decoding algorithm
of the current paper, this need not be checked separately, since if
$\sigma(X)=c\cdot \hath_1(X)$, the criterion for polynomial
evaluation from the following section will assure in particular that
the validity of $\hath_1(X)$ as an ELP will be checked. For further
details, see Subsection \ref{subsec:chien} ahead.    

\begin{remark}\label{rem:dir}
{\rm
\begin{enumerate}
\item Note that in Theorem \ref{thm:success}, if $n_2=0$, that
is, if $J$ is a ``direct hit'' including only  error locations, then it
follows from (\ref{eq:preinterp}), (\ref{eq:preinterpt}),
(\ref{eq:deltaa}), (\ref{eq:deltab}) that $f_1(X)+Yf_2(X)\in
L(J)$. Hence, the minimality of $g_0(X)+Yg_1(X)$ and
(\ref{eq:propeq}) imply that $(f_1,f_2)=c\cdot(g_0,g_1)$ for some
$c\in \efq^*$, when $f_2\neq 0$.
For speeding up the decoding,
we would also like to consider the case of an ``indirect hit'', that
is, the case where $n_1-n_2\geq\veps-t$, while $n_2>0$ in the
proposition. By (\ref{eq:mug}), in this case we may restore
$t(X^2)\sigma(X)$ from a minimizing element in $L(J)$, rather than 
$\sigma(X)$ itself. 

\item To restore $\sigma(X)$ itself, one possible method is as
follows. As $\gcd(f_1,f_2)=1$, in Theorem 
\ref{thm:success} we have $t(X)=\gcd(g_0,g_1)$. Hence, we may calculate
$t(X)$, and consequently $f_1,f_2$ from the output $g_0,g_1$. There is
also an additional method, that has a somewhat higher complexity, but
avoids the calculation of $t(X)=\gcd(g_0,g_1)$ and the division by $t(X)$
-- see Subsection \ref{subsec:chien} for more details.
\end{enumerate}
}
\end{remark}

For simplicity, from this point on we will sometimes identify a
polynomial $g_0(X)+Yg_1(X)$ with the vector $(g_0(X),g_1(X))$ without
further mention; it should be clear from the context which
representation is used.

Theorem \ref{thm:success} indicates that the decoder should look
for a minimizing element in $L(J)$, and such an element always appears
in a \grobner{} basis with respect to any monomial ordering that
resolves ties for the $(1,w)$-weighted degree, e.g., for the
ordering $<_w$ when working with the vector representation
$(g_0(X),g_1(X))$. The core idea of the \emph{fast} Chase 
decoding algorithm of the current paper is that a single application
of \kotter{}'s iteration is required for moving from
$L(\alpha_{i_1},\ldots,\alpha_{i_{r-1}})$ to
$L(\alpha_{i_1},\ldots,\alpha_{i_r})$, 
where $r\in\{1,\ldots,\eta\}$.\footnote{By convention, for $r=1$, 
$L(\alpha_{i_1},\ldots,\alpha_{i_{r-1}})=\eftwos[X]^2$.} 

Recalling that the edges of the decoding tree $T$ defined in the 
previous subsection correspond exactly to moving from a subset of
$\efq^*$ to a subset containing one additional element, we have the
following adaptation of \kotter{}'s iteration for the edges of $T$. 
Note that the root of $T$ contains the \grobner{} basis
$\{\bg_1=(1,0),\bg_2=(0,1)\}$ for $\eftwos[X]^2$.  

\hfill\\ 

\begin{tabular}{|c|}
\hline
{\bf Algorithm A: \kotter{}'s iteration on an edge of $T$}\\
\hline
\end{tabular}

\begin{description}

\item[\bf{Input}] 

\begin{itemize}

\item The weight $w:=2\deg(h_{21})-t-1/2$

\item A \grobner{} basis
$G=\{\bg_1=(g_{10},g_{11}),\bg_2=(g_{20},g_{21})\}$\\
for $L(\alpha_{i_1},\ldots,\alpha_{i_{r-1}})$ with respect to
$<_w$,\\ with $\lmw(\bg_j)$ containing the $j$-th unit vector for
$j\in\{1,2\}$  

\item $\hath_1(X)$, $\hath_2(X)$ 

\item The next error location, $\alpha_{i_r}$, with
$i_r\notin\{i_1,\ldots,i_{r-1}\}$ 

\end{itemize}

\item[\bf{Output}] A \grobner{} basis
$\Gplus=\{\bgplus_1=(\gplus_{10},\gplus_{11}),
\bgplus_2=(\gplus_{20},\gplus_{21})\}$ 
for \\ $L(\alpha_{i_1},\ldots,\alpha_{i_{r-1}},\alpha_{i_r})$ with 
respect to $<_w$,\\ 
with $\lmw(\bgplus_j)$ containing the $j$-th unit vector for
$j\in\{1,2\}$  

\item[\bf{Algorithm}]

\begin{itemize} 

\item For $j=1, 2$, calculate /* justification: (\ref{eq:deltaa}),
(\ref{eq:deltab}) */
$$
\Delta_j:=\begin{cases}
\displaystyle
g_{j0}(\alpha_{i_r}^{-2}) +
\frac{\hath_2(\alpha_{i_r}^{-1})}{\hath_1(\alpha_{i_r}^{-1})}
g_{j1}(\alpha_{i_r}^{-2}) & \text{if } \hath_1(\alpha_{i_r}^{-1})\neq 0\\
g_{j1}(\alpha_{i_r}^{-2}) & \text{if }\hath_1(\alpha_{i_r}^{-1})= 0 
\end{cases}
$$

\item Set $A:=\big\{j\in\{1,2\}|\Delta_j\neq 0\big\}$ 

\item For $j\in\{1,2\}\smallsetminus A$, set $\bgplus_j:=\bg_j$

\item Let $j^*\in A$ be such that  $\lmw(\bg_{j^*})=\min_{j\in
A}\{\lmw(\bg_j)\}$  

\item For $j\in A$ 

\begin{itemize}

\item If $j\neq j^*$

\begin{itemize}

\item Set 
$\bgplus_j := \frac{\Delta_{j^*}}{\Delta_j}\bg_j+ \bg_{j^*}$ 

\end{itemize}

\item Else /* $j=j^*$ */

\begin{itemize}

\item Set $\bgplus_{j^*} := (X+\alpha_{i_r}^{-2})\bg_{j^*}$

\end{itemize}

\end{itemize}

\end{itemize}

\end{description}

Note that the update for $j=j^*$ is justified by the fact that
replacing $g_{j^*}$ by $Xg_{j^*}$ has the effect of multiplying
$\Delta_{j^*}$ by $\alpha_{i_r}^{-2}$. Note also that the update
$\bgplus_j := \frac{\Delta_{j^*}}{\Delta_j}\bg_j+ \bg_{j^*}$ (which is 
equivalent to the usual update of the form $\bgplus_j :=
\bg_j+\frac{\Delta_j}{\Delta_{j^*}} \bg_{j^*}$) appears in this form
only for the purpose of the complexity analysis ahead. This update
rule assures that when $|A|=2$, both $\bg_j$ and $\bg_{j^*}$ are
multiplied once by a constant, where $j:=\{1,2\}\smallsetminus j^*$.

\begin{remark}\label{rem:alga}
{\rm
Some remarks regarding Algorithm A are in order.
\begin{enumerate}

\item By Remark \ref{rem:lj}, it is clear that there is a strict
inclusion 
$$
L(\alpha_{i_1},\ldots,\alpha_{i_{r-1}}) \supsetneq
L(\alpha_{i_1},\ldots,\alpha_{i_{r-1}},\alpha_{i_r}).  
$$
In some detail,
$v(X):=X^{2t}\cdot\prod_{j=1}^{r-1}(X^2+\alpha_{i_j}^{-2})$ is 
in the solution module $M$, and for the unique $g_0(X),g_1(X)$ such
that $\mu^{-1}(v)=g_0\bh_1+g_1\bh_2$, it holds that
$g_0+Yg_1\in L(\alpha_{i_1},\ldots,\alpha_{i_{r-1}})$, but
$g_0+Yg_1\notin L(\alpha_{i_1},\ldots,\alpha_{i_{r}})$.  
Hence, in Algorithm A, $\Delta_j$ can be zero for at most one value
of $j$, for otherwise the above two modules would share a
\grobner{} basis.

\item The evaluations $\hath_1(\alpha_{i_j}^{-1}),
\hath_2(\alpha_{i_j}^{-1})$ on
\emph{all} $\eta$ unreliable coordinates can be calculated in advance
before starting the fast Chase algorithm, and so can the quotients
$\frac{\hath_2(\alpha_{i_j}^{-1})}{\hath_1(\alpha_{i_j}^{-1})}$ (when
$\hath_1(\alpha_{i_j}^{-1})\neq 0$), and the values $\alpha_{i_j}^{-2}$
to be substituted in various polynomials. 

\item Actually, it is not essential to work with a
non-integer weight, and we could have chosen $w=2\deg(h_{21})-t-1$,
i.e., the lower of the two values in (\ref{eq:deltaw}). The reason
that this value of $w$ works also for the odd case (where we should
use $w+1$ instead) is the following fact: For integer $w$, the orders
$<_w$ and $<_{w+1}$ differ only in the way they compare monomials that
have the same $(1,w+1)$-weighted degree (see Appendix \ref{app:onew}
for details). Hence, if we only wish to minimize the
$(1,w+1)$-weighted degree, we can use $<_w$ instead of
$<_{w+1}$.\footnote{
It is \emph{not} true that the order $<_w$ can be used to 
minimize the $(1,w')$-weighted degree for \emph{all} $w'\geq w$. A
vector $\bs{h}$ in a submodule $P\subseteq K[X]^2$ (for a field $K$)
with a minimal $<_w$-leading monomial minimizes both the
$(1,w)$-weighted degree and the $(1,w+1)$-weighted degree, but not
necessarily the $(1,w+2)$-weighted degree, as it need not be
$<_{w+1}$-minimal. For example, considering $P:=XK[X]^2$, 
$(0,X)$ is the unique monomial minimizing the $(1,-1)$-weighted
degree, and both $(X,0)<_0 (0,X)$ 
have the minimal $(1,0)$-weighted degree, while $(X,0)$ has the
minimal $(1,1)$-weighted degree. 
}
Note that this observation is relevant only for $\rho=1$; in the more 
general context of list decoding, we still have to use a non-integer $w$.

\end{enumerate}
}
\end{remark}

\subsection{A criterion for polynomial evaluation, and efficient
evaluation}\label{subsec:chien}

\subsubsection{The stopping criterion}\label{subsec:stop}
In \cite{Wu12}, Wu suggested a stopping criterion for avoiding
unnecessary exhaustive substitution in polynomials. Wu's criterion is
based on the idea that a length variable in his algorithm does not
increase if the true ELP was found, and an additional flip hits a
correct error location. Here, we have a similar stopping criterion
based on the discrepancy of Algorithm A, in the lines of the criterion
in \cite{SB21}.

In detail, using the terminology of Theorem
\ref{thm:success}, suppose that for
$J=\{\alpha_{i_1},\ldots,\alpha_{i_r}\}$ it holds that $n_1-n_2\geq
\veps-t+1$, and that $\alpha_{i_r}$ is an error location. Then for
$J'=\{\alpha_{i_1},\ldots,\alpha_{i_{r-1}}\}$ we have $n_1-n_2\geq
\veps-t$, and the condition of Theorem \ref{thm:success} holds.  

By the theorem, if $\hath_1$ is not the 
ELP (up to a multiplicative constant), then after performing Algorithm
A on an edge reaching the vertex 
$(\alpha_{i_1},\ldots,\alpha_{i_{r-1}})$, we have $\mu(g^+_{j 0}\bh_1
+g^+_{j 1}\bh_2)=t(X^2)\sigma(X)$ for the unique $j\in\{1,2\}$ for
which $\wdegw(\bgplus_j)$ is minimal.\footnote{
Since $w$ is non-integer, we must have $\wdegw(\bgplus_1)\neq
\wdegw(\bgplus_2)$. In detail, since the leading monomial of
$\bgplus_2$ is on the right, $\wdegw(\bgplus_2)$ is non-integer. On
the other hand, the leading monomial of $\bgplus_1$ is on the left, and
therefore $\wdegw(\bgplus_1)$ is an integer.
}
Hence, on the edge connecting
$(\alpha_{i_1},\ldots,\alpha_{i_{r-1}})$ to
$(\alpha_{i_1},\ldots,\alpha_{i_{r}})$ it holds that $\Delta_j=0$, by
Remark \ref{rem:root} and the assumption that $\alpha_{i_r}$ is an
error location. 

Moreover, if $\hath_1=c\cdot \sigma(X)$ for some $c\neq 0$, suppose
that at least one of the 
$\eta$ unreliable coordinates, say, $\alpha_{i_1}$, is an error
location. As the root of the tree $T$ contains the \grobner{} basis
$\{(1,0),(0,1)\}$ of $\efq[X]^2$, and
$\mu(1\cdot\bh_1+0\cdot\bh_2)=\hath_1$, it holds that $\Delta_1=0$ on
the edge connecting the root to the vertex $(\alpha_{i_1})$. 

We conclude that when moving from the root of the tree to depth $1$,
the criterion for starting an exhaustive substitution is having
$\Delta_1=0$, while for $r>1$, the criterion for exhaustive
substitution on an edge connecting depth $r-1$ to depth $r$ is
$\Delta_j=0$ for the unique $j\in\{1,2\}$ for
which $\wdegw(\bg_j)$ is minimal.

if $\veps\leq t+r$  for
$r+1\leq \rmax$, and $r+1$ of the errors are in the $\eta$ unreliable
coordinates, then the required exhaustive substitution will never be
missed. Note that this slightly degrades the decoding performance, as
we need $r+1$ errors on the unreliable coordinates, instead of the $r$
required to restore the ELP on some vertex.  

It is possible that the stopping
criterion holds falsely. Heuristically, it is reasonable to assume
that the probability that one discrepancy is $0$ by accident is about
$1/q$.\footnote{
We remark that the event of meeting the stopping criterion on the
edge connecting $(\alpha_{i_1},\ldots,\alpha_{i_{r-1}})$ to
$(\alpha_{i_1},\ldots,\alpha_{i_{r-1}},\alpha_{i_r})$ is exactly the
event that the minimal element in
$L(\alpha_{i_1},\ldots,\alpha_{i_{r-1}})$ (with respect to $\lmw$)
happens to fall in
$L(\alpha_{i_1},\ldots,\alpha_{i_{r-1}},\alpha_{i_r})$.    
} When this happens, an unnecessary exhaustive substitution is
performed. This somewhat increases the decoding complexity, but has no
effect on the probability of decoding successfully.

\begin{examplenn}
{\rm
For $q=256$ ($n=255$) and $t=8$, a random error of weight $\veps=14$ was
drawn $10^4$ times. Algorithm A was run on a path corresponding to
$\veps-t=6$ totally random distinct positions, and the number of times
for which the stopping criterion was falsely met was counted. The
resulting empirical probability was calculated as the total number of
false positives on an edge, divided by $6\cdot 10^4$, and was equal
to $1/251.05$.  
}
\end{examplenn}

\subsubsection{Efficient evaluation by finding $t(X)$ and division}
\label{subsec:chient}
Let $(g_0,g_1)$ be the
vector with the minimum $(1,w)$-weighted degree in the current
\grobner{} basis. As shown in Remark \ref{rem:dir}, in case of success,
we have $(g_0,g_1)=t(X)(f_1,f_2)$, and $t(X)=\gcd(g_0,g_1)$. This
leads to the following evaluation strategy:

\begin{itemize}

\item Before starting the fast Chase decoding algorithm, calculate and
store $\{\hath_1(x)\}_{x\in\efq^*},\{\hath_2(x)\}_{x\in\efq^*}$ 

\item When traversing the decoding tree $T$, if the stopping criterion
holds, 

\begin{enumerate}

\item Calculate $t(X)=\gcd(g_0,g_1)$ using the Euclidean
algorithm. This requires $O(r^2)$ multiplications (by Proposition
\ref{prop:degs} ahead). 

\item Calculate $f_1(X):=g_0(X)/t(X)$, $f_2(X):=g_1(X)/t(X)$. This
also requires $O(r^2)$ multiplications.

\item For all $x\in \efq^*$, 

\begin{itemize}

\item Calculate $a:=f_1(x^{-2}),b:=f_2(x^{-2})$, 

\item Read the stored values $c:=\hath_1(x^{-1}),d:=\hath_2(x^{-1})$.

\item Calculate $ac+bd$. If the result is $0$, then adjoin $x$ to a
set $E$ of error locations. 

\end{itemize}

\item Calculate
$\delta:=\deg(f_1(X^2)\hath_1(X)+f_2(X^2)\hath_2(X))$ as 
$$
\delta = \max\big\{2\deg(f_1)+2\deg(h_{10})+1,
2\deg(f_2)+2\deg(h_{21})\big\}
$$
(by Proposition \ref{prop:sigdeg} ahead) and check if 
$\delta=|E|$. If equality holds, adjoin the error vector with
support $E$ to the list of output error vectors (see justification
ahead).  

\end{enumerate}

\end{itemize}

Note that the comparison of the degree of the estimated ELP to the
number of its roots is equivalent to checking that the syndrome of the
error vector defined by the estimated ELP is
equal to the syndrome of $\bs{y}$: First, the estimated ELP
$\tilde{\sigma}:=\mu(f_1\bh_1+f_2\bh_2)$ is separable, by the degree
test. Also, since $\tilde{\sigma}\in \mu(N)=M$, we may use
\cite[Prop.~4.3]{SB21} with $\tilde{\omega}:=\tilde{\sigma}'$ to show
that the error vector with locations defined by $\tilde{\sigma}$ and
values obtained by $\tilde{\sigma}$ and
$\tilde{\omega}=\tilde{\sigma}'$ through {\it Forney's formula}
(see e.g., \cite[Eq.~(3)]{SB21}) has the same syndrome as 
$\bs{y}$. Finally, observing Forney's formula for the case where
$\omega=\sigma'$ and $a_i=\alpha_i$ for all $i$ (using the terminology
of \cite{SB21}), as is the case for BCH codes, we see that all the
error values obtained with this formula must be $1$.

It remains to justify the calculation of the degree $\delta$ in
the above steps. For this, we have the following proposition.
\begin{proposition}\label{prop:sigdeg}
For any polynomials $u_1(X),u_2(X)\in \eftwos[X]$, it holds that 
\begin{multline*}
\deg\big(u_1(X^2)\hath_1(X)+u_2(X^2)\hath_2(X)\big) =\\
\max\big\{2\deg(u_1)+2\deg(h_{10})+1, 2\deg(u_2) + 2\deg(h_{21})\big\}.
\end{multline*}
\end{proposition}
\begin{proof}
Recall that $\hath_j(X)=h_{j1}(X^2)+Xh_{j0}(X^2)$ ($j=1,2$). Consider
the degree of $h_{11}(X^2)+Xh_{10}(X^2)$. As the leading monomial of
$\bh_1(X)$ with respect to $<_{-1}$ is on the left, we have
$\deg(h_{10})>\deg(h_{11})-1$, that is, $\deg(h_{10})\geq
\deg(h_{11})$. Hence  
$$
\deg(\hath_1)=\deg(h_{11}(X^2)+Xh_{10}(X^2))= 1+2\deg(h_{10}).
$$     
Similarly, as the leading monomial of $\bh_2$ is on the right, we have
$$
\deg(\hath_2)=\deg(h_{21}(X^2)+Xh_{20}(X^2))= 2\deg(h_{21}).
$$
Now the assertion follows from the fact that the degrees of
$u_1(X^2)\hath_1(X)$ and $u_2(X^2)\hath_2(X)$ are
distinct, as the first is odd, while the second is even.
\end{proof}

\subsubsection{Efficient evaluation without polynomial division}
\label{subsec:chiennot}
Let us now consider another efficient method to perform the exhaustive
evaluation. The current method is slightly less efficient than that of
the previous subsection, but it avoids polynomial division and gcd
calculations (Euclidean algorithm), which may be desirable in some
applications. Letting again $(g_0,g_1)$ be the vector with the minimum 
$(1,w)$-weighted degree in the current \grobner{} basis and writing
\begin{equation}\label{eq:hsig}
\hat{\sigma}(X):= \mu(g_0\bh_1 + g_1\bh_2) = g_0(X^2)\hath_1(X) +
g_1(X^2)\hath_2(X),  
\end{equation}
it follows from Theorem \ref{thm:success} that in case of success,
$\hat{\sigma}(X)=t(X^2)\sigma(X)$ for some non-zero $t(X)$ such that
$\gcd(t(X^2),\sigma(X))=1$ and $t(X^2)$ has no roots outside
$J^{-1}:=\{\alpha_{i_1}^{-1},\ldots,\alpha_{i_{r-1}}^{-1}\}$ (using the
terminology of Subsection \ref{subsec:stop}). Noting that
$\hat{\sigma}'(X)=t(X^2)\sigma'(X)$ and recalling that
$\sigma,\sigma'$ are coprime, it follows that the roots of $\sigma(X)$
are exactly those roots of $\hat{\sigma}(X)$ that are not roots of
$\hat{\sigma}'(X)$. Also, since $t(X^2)$ has no roots outside
$J^{-1}$, all the roots of $\hat{\sigma}$ outside $J^{-1}$ are also
roots of $\sigma$.
As 
$$
\hat{\sigma}'(X)=g_0(X^2)\hath'_1(X) + g_1(X^2)\hath'_2(X),
$$ 
this leads to the following evaluation strategy:

\begin{itemize}

\item Before starting the fast Chase decoding algorithm, calculate and
store 
$$
\{\hath_1(x)\}_{x\in\efq^*},\{\hath_2(x)\}_{x\in\efq^*},
\{\hath'_1(x)\}_{x\in\efq^*}, \{\hath'_2(x)\}_{x\in\efq^*}.
$$ 

\item When traversing the decoding tree $T$, if the stopping criterion
of Subsection \ref{subsec:stop} holds:

\begin{enumerate}

\item Find the estimated error locations by performing the following
actions for all $x\in \efq^*$:
\begin{itemize}

\item Calculate $a:=g_0(x^{-2}),b:=g_1(x^{-2})$.

\item Read the stored values $c:=\hath_1(x^{-1})$ and
$d:=\hath_2(x^{-1})$ and calculate $\hat{\sigma}(x^{-1}) = ac + bd$.

\item If $x^{-1}\in J^{-1}$,

\begin{itemize}

\item Read the stored values $c':=\hath'_1(x^{-1})$ and
$d':=\hath'_2(x^{-1})$ and calculate $\hat{\sigma}'(x^{-1}) = ac' +
bd'$.

\item If $\hat{\sigma}(x^{-1})=0$ and $\hat{\sigma}'(x^{-1})\neq 0$,
adjoin $x$ to a set $E$ of error locations.

\end{itemize}

\item Otherwise ($x^{-1}\notin J^{-1}$)

\begin{itemize}

\item If $\hat{\sigma}(x^{-1})=0$ adjoin $x$ to a set $E$ of error
locations. 

\end{itemize}

\end{itemize}  

\item Check if the potential error vector with support $E$ has the
same syndrome as $\bs{y}$. If it does, adjoin it to the output error
list of the decoder. This requires $t|E|$ additions in
$\eftwos$ (no multiplications in $\eftwos$ are required). 

\end{enumerate}

\end{itemize}

With this method, we do not know how to avoid
calculating the syndrome of the error vector for validation. Also, the
evaluated polynomials  
$g_0,g_1$ here typically have a higher degree than the polynomials
evaluated in the previous subsection.  The degree is higher
by $d_g:=\deg(\gcd(g_0,g_1))$, and the overall complexity increment of evaluation is
therefore $O(d_g \cdot n)$ (if the FFT of \cite{GM10} is not used). If
$d_g>0$, this is typically higher than the complexity of $O(r^2)$ 
required in the previous subsection for finding $t(X)$ and
division.\footnote{Note that when \cite{GM10} is not used for
evaluation, it is reasonable to assume that $r=O(\log(n))$, for
otherwise evaluation with \cite{GM10} is more efficient.} 

\subsection{Complexity analysis}\label{subsec:complexity}
Next, we would like to evaluate the complexity of Algorithm A. To do this,
we first have to bound the degrees of the updated polynomials.

\begin{definition}
{\rm
For $r\in\{1,\ldots,\rmax\}$ and $j\in \{1,2\}$, write
$\ared{\bg_j}=(\ared{g_{j0}},\ared{g_{j1}})$ and
$\ared{\bgplus_j}=(\ared{\gplus_{j0}},\ared{\gplus_{j1}})$ for the respective
polynomials in Algorithm $A$ when adjoining error location
$\alpha_{i_r}$. By definition,
$\varared{\bg_1}{1}=(1,0)$ and $\varared{\bg_2}{1}=(0,1)$ are the
initial values used in the first call to the algorithm. We will use
the convention $\varared{\bgplus_1}{0}:=(1,0)$ and
$\varared{\bgplus_2}{0}:=(0,1)$. 
}
\end{definition}

\begin{proposition}\label{prop:degs}
For all $r$,
\begin{equation}\label{eq:degsum}
\deg(\ared{\gplus_{10}}) + \deg(\ared{\gplus_{11}}) +
\deg(\ared{\gplus_{20}}) + \deg(\ared{\gplus_{21}}) \leq 2r-1, 
\end{equation}
where the sum on the left-hand side is taken only over non-zero
polynomials. 
\end{proposition}

\begin{proof}
For simplicity, let us define the degree of monomials in
$\eftwos[X]^2$ in the obvious way, by setting $\deg(X^k,0):=k$ and
$\deg(0,X^k):=k$.  
In each application of \kotter{}'s iteration, there is at most one
$j\in\{1,2\}$ for which $\lm(\bgplus_j)>\lm(\bg_j)$, namely
$j=j^*$. Moreover, 
$\lm(\bgplus_{j^*})=X\lm(\bg_{j^*})$, so that
$\deg\lm(\bgplus_{j^*})=\deg\lm(\bg_{j^*})+1$. Also, in the
initialization we have
$\deg\lm(\varared{\bg_1}{1})=\deg\lm(\varared{\bg_2}{1})=0$. Hence, 
for all $r$, 
\begin{equation}\label{eq:deglm}
\deg\lm(\ared{\bgplus_1})+\deg\lm(\ared{\bgplus_2})\leq r.
\end{equation}

Recall that our monomial ordering is 
$<_w$, and that in \kotter{}'s iteration, the leading
monomial of $\bgplus_j$ contains the $j$-th unit
vector. Hence 
\begin{equation}\label{eq:deglmone}
\deg\lmw(\ared{\bgplus_1})=\deg(\ared{\gplus_{10}}) >
\deg(\ared{\gplus_{11}})+w 
\end{equation}
and
\begin{equation}\label{eq:deglmtwo}
\deg\lmw(\ared{\bgplus_2}) = \deg(\ared{\gplus_{21}}) \geq
\deg(\ared{\gplus_{20}})-w.
\end{equation}
From (\ref{eq:deglm}) and the equality part of (\ref{eq:deglmone}) and
(\ref{eq:deglmtwo}), we get 
\begin{equation}\label{eq:degnonz}
\deg(\ared{\gplus_{10}})+\deg(\ared{\gplus_{21}}) \leq  r. 
\end{equation}
Suppose that both $\ared{\gplus_{11}}$ and $\ared{\gplus_{20}}$ are 
non-zero (note that $\ared{\gplus_{10}}$ and $\ared{\gplus_{21}}$,
containing the respective leading monomials, are never zero). Then from
the inequality part of (\ref{eq:deglmone}) and (\ref{eq:deglmtwo}), we
see that  
$$
\deg(\ared{\gplus_{11}})+\deg(\ared{\gplus_{20}})<
\deg(\ared{\gplus_{10}})+\deg(\ared{\gplus_{21}}) \leq r.  
$$
Summing two inequalities, we get (\ref{eq:degsum}).

Now, if $\ared{\gplus_{11}}=\ared{\gplus_{20}}=0$, then the assertion
readily follows from (\ref{eq:degnonz}), and so it remains to
consider only the case where exactly one of
$\ared{\gplus_{11}},\ared{\gplus_{20}}$ is zero, say,
$\ared{\gplus_{11}}=0$ and $\ared{\gplus_{20}}\neq 0$. We will use
induction on $r$ to prove that $\deg(\ared{\gplus_{20}})\leq r-1$. The
basis of the induction, on the root of the $T$, is clear. Suppose that
the assumption holds for $r-1$, and recall that
$\ared{\bg_j}=\varared{\bgplus_j}{r-1}$ for $j=1,2$. If $\Delta_2=0$
in iteration $r$, then there is nothing to prove. Suppose,
therefore that $\Delta_2\neq 0$. If $j^*=2$, then the assertion follows
immediately from the update rule for $j^*$ and the induction
hypothesis. Otherwise, since $\deg(\ared{g_{10}})\leq r-1$ by
(\ref{eq:degnonz}), the assertion follows by the update rule for
$j\neq j^*$ and the induction hypothesis. 
\end{proof}

Now, on an edge connecting depth $r-1$ to depth $r$ in the decoding
tree $T$, we have four polynomials
$\ared{g_{ij}}=\armoed{\gplus_{ij}}$, whose sums of degrees is at 
most $2(r-1)-1=2r-3$, by Proposition \ref{prop:degs}. For each one of
these polynomials, we have to 
perform evaluation once, and multiply by a scalar once. We also have
to calculate $\Delta_{j^*}/\Delta_j$, at the cost of a single
multiplication, assuming the inverse is calculated by a
table. Finally, there are $2$ multiplications by the pre-computed
$\hath_2(\alpha_{i_r}^{-1})/\hath_1(\alpha_{i_r}^{-1})$. In
general, for a polynomial $f$, evaluation takes $\deg(f)$
multiplications (using {\it Horner's method}), and multiplying by a
scalar takes $\deg(f)+1$ multiplications. Hence, the overall number of
multiplications for performing Algorithm A on an edge between depth
$r-1$ and $r$ is $3+2(2r-3)+4=4r+1$. In comparison, the corresponding 
complexity for \cite[Alg.~C]{SB21} is $20r+3$ (this is
``$M_C$'' from \cite{SB21}). 

Hence, the gain from moving to the binary alphabet is by
a factor of about $5$. We note that having a single \kotter{}
iteration per edge effects the complexity twice, both in halving
the degrees of the maintained polynomials, and in halving the number
of evaluations and multiplications in each stage. The reason that the
gain in comparison to the $q$-ary case is by a factor of $5$ (instead
of $4$), is that the so-called {\it derivate step} of
\cite[Alg.~C]{SB21}, which does not have an equivalent in 
Algorithm A, is slightly more complicated than the {\it root step}.

When $\rmax=\eta$, if the entire tree
$T$ is traversed, the total number of multiplications for all edge
updates is at most  
$$
\sum_{r=1}^{\eta}(4r+1)\binom{\eta}{r} =
4\sum_{r=1}^{\eta}\eta\binom{\eta-1}{r-1}+2^{\eta}-1 
= \eta 2^{\eta+1} +2^{\eta}-1
$$
(using $r\binom{\eta}{r}=\eta\binom{\eta-1}{r-1}$).

A complexity analysis of Wu's algorithm for binary BCH codes
\cite[Alg.~5]{Wu12} is missing from \cite{Wu12}. While a precise
estimation of the complexity of Wu's algorithm is outside the scope of
the current paper, we note that it is initialized by two polynomials
whose sum of degrees is $2t+1$, and on each edge, typically the sum of
degrees is increased by $2$. Hence, it seems
reasonable to replace the total degree bound of Proposition
\ref{prop:degs} by about $2(t+r)+1$ for Wu's algorithm. 

The polynomials participating on
an edge connecting a vertex at depth $r-1$ with an edge at depth $r$
therefore have a sum of degrees of $2(t+r-1)+1=2(t+r)-1$. Hence, the
evaluation part on such an edge requires $2(t+r)-1$ multiplications, while
the multiplication of polynomials by constants takes $2(t+r)+1$
multiplications (adding $2$ to account, as before, for $2$ free
coefficients). An additional single multiplication is required to
calculate a ratio of two scalars, as above. The overall number of
multiplication for such an edge is therefore $4(t+r)+1$. 

When $\rmax=\eta$, if the entire tree
$T$ is traversed, the total number of multiplications for all edge
updates is at most
$\sum_{r=1}^{\eta}(4(t+r)+1)\binom{\eta}{r}=2^{\eta+1}(\eta+2t+1/2)-4t-1$.
It follows that the algorithm of the current paper reduces the complexity
by a factor of about $(\eta+2t)/\eta=1+2t/\eta$ when $\rmax=\eta$ and
the entire tree is traversed. The complexity reduction is even higher
when $\rmax<\eta$.  

Regarding \cite{K01}, we note that from the discussion on p.~2004 of
the precursor \cite{K99}, it follows that the complexity of
polynomial updates for the fast Chase decoding of \cite{K01} is in
$O(2^{\eta+1}t)$. While it is not entirely clear what is the precise
constant involved, it seems that the overall complexity is higher than
that of \cite{Wu12} (see \cite{ZZW11}). 

Let us now consider the cost of unnecessary polynomial
evaluations, focusing only the method of Subsection
\ref{subsec:chient}. First, it can be
verified from (\ref{eq:degnonz}) that, using the terminology of
Subsection \ref{subsec:chient}, $\deg(g_0)+\deg(g_1)\leq
r-1$ on an edge connecting depth $r$ to depth $r+1$.\footnote{To see
this, we note that if $\lmw(\bg_1)<_w\lmw(\bg_2)$, then
$\deg(g_{11})+w<\deg(g_{10})\leq\deg(g_{21})+w$ (using an obvious 
notation), so that, $\deg(g_{10})+\deg(g_{11}) <
\deg(g_{10})+\deg(g_{21})\leq r$, where the last inequality is from
(\ref{eq:degnonz}). Similarly, if $\lmw(\bg_2)<_w\lmw(\bg_1)$, then
$\deg(g_{20})\leq \deg(g_{21})+w<\deg(g_{10})$, so that
$\deg(g_{20})+\deg(g_{21})<r$.} Hence, on an edge connecting depth
$r-1$ to depth $r$, we have $\deg(g_0)+\deg(g_1)\leq
r-2$.

The number of multiplications for the gcd calculation and the divisions
from Subsection \ref{subsec:chient} is therefore $O(r^2)$ (see, e.g.,
\cite[Thm.~17.3]{Shoup}). Evaluation of two polynomials whose sum of
degrees is at most $r-2$ (for $r\geq 2$) can be done using
$n\cdot\min\{4\log_2(n),r-2\}$ 
multiplications, where the minimum is between point-by-point
evaluation using Horner's method and the FFT algorithm of
\cite{GM10}. In addition, there are $2n$ multiplications in 
calculations of the form $ac+bd$.

As already mentioned, it seems reasonable to
assume that the probability of falsely meeting the stopping condition
of Subsection \ref{subsec:stop} is about $1/q$. With this assumption,
the mean contribution per-edge of unnecessary exhaustive evaluations
is 
$$
\frac{1}{q}\big(O(r^2)+2n+n\cdot\min\{4\log_2(n),r-2\}\big),
$$
which is dominated by
$\bar{N}_{\mathrm{eval}}:=2+\min\{4\log_2(n),r-2\}$. When the minimum
is $r-2$, this increases the overall mean number of multiplications by
a factor of about $5/4$ over the number $4r+1$ solely for polynomial
updates. 

Comparing to Wu's algorithm, recall first that the approximated
probability of $1/q$ applies also to the stopping condition of
\cite{Wu12}. Hence, arguing as above, the mean contribution
per-edge of unnecessary exhaustive evaluations in \cite{Wu12} is
about $\min\{2\log_2(n),t+r\}$ (note that for \cite{Wu12}, only a single
polynomial has to be evaluated). Again, when the minimum is $t+r$,
this increases the overall mean number of multiplications by a factor
of about $5/4$ over polynomial evaluation.

Finally, we note that the
complexity of the pre-computations in part 2 of Remark \ref{rem:alga}
is in $O(\eta d)$, and hence negligible.
\subsection{High-level description of the decoding process}
Let us now wrap-up the entire decoding process with the fast Chase
algorithm of the current paper.

\begin{enumerate}

\item Perform HD decoding in the following way:

\begin{itemize}

\item Calculate the syndrome polynomial $S(X)$.

\item Calculate the modified syndrome
$\hat{S}(X)=\frac{\even{S}(X)}{1+X\odd{S}(X)}\mod(X^t)$ using the
recursion (\ref{eq:recur}).

\item Find a \grobner{} basis 
$\{\bh_1=(h_{10},h_{11}),\bh_2=(h_{20},h_{21})\}$ with respect to
$<_{-1}$ for the module $N$ of Definition \ref{def:solmod}, where  
the leading monomial of $\bh_1$ is in the first coordinate and the
leading monomial of $\bh_2$ in the second coordinate. As explained in
\cite{Fitz95}, this can be practically done with any of the
syndrome-based algorithms for HD decoding of RS codes.

\item Let $j$ be such
$\lmmo(\bh_j)=\min\{\lmmo(\bh_1),\lmmo(\bh_2)\}$. Take
$\sigma:=\mu(\bh_j)=h_{j1}(X^2) + Xh_{j0}(X^2)$

\item If $\deg(\sigma)\leq t$, find the roots of $\sigma$ in $\efq^*$
by exhaustive substitution. If the number of roots equals
$\deg(\sigma)$, declare success and output the corresponding error vector. 

\end{itemize}

\item If HD decoding fails, use reliability information to identify
the set $U$ of $\eta$ least reliable coordinates. calculate
$w:=2\deg(h_{21})-t-1/2$, store the basis $\{(1,0),(0,1)\}$ at the
memory for depth $0$, and 
perform fast Chase decoding by traversing the tree $T$, depth first:

\begin{itemize}

\item When visiting the edge
$(\bs{\beta}',\bs{\beta})$ between  a vertex  $\bs{\beta}'$ at depth
$r-1$ and a 
vertex $\bs{\beta}$ at depth $r$ with an additional $1$ in coordinate
$\alpha_{i_r}$:   

\begin{itemize}

\item Perform Algorithm A, taking the inputs $\bg_1,\bg_2$ from
the memory for depth $r-1$. If the stopping criterion holds, perform
efficient exhaustive substitution, in one of the methods described in
Subsection \ref{subsec:chient} and Subsection
\ref{subsec:chiennot}. If the resulting error passes the
verification criterion described in the 
respective subsections, adjoin the error to the list of output errors. 

\item Store the outputs $\bgplus_1,\bgplus_2$ in the memory for depth
$r$.  

\end{itemize}

\end{itemize} 

\end{enumerate}

\section{Conclusions}\label{sec:conclusions}
Building on the SD Wu list decoding algorithm, we have presented a new
syndrome-based fast Chase decoding algorithm for 
binary BCH codes. The new algorithm requires only a single
\kotter{} iteration per edge of the decoding tree, as opposed to two
iterations in the corresponding algorithm
for RS codes \cite{SB21}. Also, the algorithm has a lower complexity
than that of \cite{Wu12}, as it updates low-degree coefficient
polynomials. As in \cite{SB21}, and as opposed to \cite{Wu12}, the
current fast Chase algorithm can work also if the total number of
errors is beyond $d-1$. 

We have also established an isomorphism between two solution modules
for decoding binary BCH codes. This isomorphism can be used to benefit
from the binary alphabet for reducing the complexity of HD
bounded-distance decoding in a systematic way, practically for all
syndrome-based HD decoding algorithms.   

The new fast Chase algorithm is based on the idea that for a maximum
list size  of 1, $\{0,1\}$-multiplicity assignment in the SD Wu list
decoding algorithm is equivalent to flipping bits in locations with a
non-zero multiplicity (Theorem \ref{thm:success}). This allows to
accumulate flippings through \kotter{}'s iteration, and leads to the
fast Chase decoding algorithm.

We note that the BM algorithm itself can be used to
find the \grobner{} basis $\{\bh_1,\bh_2\}$ for the module $N$. The idea
is that two vectors are extracted from the output of the algorithm, and
then at most a single leading monomial cancellation is required for
achieving the required \grobner{} basis. In fact, the calculation of
$\hath_1,\hath_2$ can be done without calculating the modified syndrome
at all, using a slightly augmented BM algorithm: In addition
to tracking the {\it connection polynomial} before the last length
change, one tracks one additional past version of the connection
polynomial. Since the proofs are rather long and this is outside the
main scope of the current paper, we omit further details. For a
precise listing of the relevant algorithms, see \cite{SK19}.

It seems plausible that the fast Chase
decoding algorithm of the current paper can also be derived through a
suitable minimization problem over a module, as done in \cite{SB21}
for RS codes. It would be interesting to formalize and prove this. It would
be even more interesting to find a way to obtain the results of
\cite{SB21} through the SD Wu list decoding algorithm for RS
codes. Beyond the interest in finding new connections between decoding
algorithms, there are also additional advantages in using the Wu list
decoding as a means for deriving a fast Chase algorithm. For example,
the case of indirect hits (Remark \ref{rem:dir}) is handled trivially,
as opposed to quite some effort in \cite[Appendix A]{SB21}.

\section*{Acknowledgment}
We thank Avner Dor and Itzhak Tamo for some very helpful
discussions. We also thank Johan Rosenkilde for pointing us to
\cite{BL94}, \cite{BL97}, \cite{N16}, and \cite{RS21}.

\appendix

\section{Proof of part 4 of Proposition \ref{cw_prop:yzhomog}} 
\label{app:homog} 
\begin{proof}
Writing $f$ as a sum of the form $\sum_{i_1,i_2,i_3}
b_{i_1,i_2,i_3}X^{i_1}Y^{i_2}Z^{i_3}$ with
$b_{i_1,i_2,\rho-i_2}:=a_{i_1,i_2}$ and with $b_{i_1,i_2,i_3}:=0$ for 
$i_3\neq \rho-i_2$, it follows Proposition \ref{cw_prop:hassecalc}
that for all $(j_1,j_2,j_3)$ and for all $\alpha\in K$
\begin{eqnarray*}
\hasse{f}{j_1,j_2,j_3}(x_0,\alpha y_0,\alpha z_0) & = & \sum_{\substack{i_1\geq
j_1,i_2\geq j_2 \\i_3=\rho-i_2\geq j_3}}b_{i_1,i_2,i_3}
\binom{i_1}{j_1}\binom{i_2}{j_2}\binom{i_3}{j_3}  x_0^{i_1-j_1}(\alpha
y_0)^{i_2-j_2}(\alpha z_0)^{\rho-i_2-j_3}\\
& = & \alpha^{\rho-j_2-j_3}\hasse{f}{j_1,j_2,j_3}(x_0,y_0,z_0).
\end{eqnarray*}
Note that as $\alpha$ is always raised to a non-negative power,
the above indeed holds also for $\alpha=0$. For this choice of
$\alpha$, we obtain
$$
\hasse{f}{j_1,j_2,j_3}(x_0,y_0,z_0) = 0\quad \implies\quad
\hasse{f}{j_1,j_2,j_3}(x_0,0,0) =0,
$$
which is sufficient for proving the assertion.\footnote{It also follows
that $\hasse{f}{j_1,j_2,j_3}(x_0,0,0)$  is non-zero only if
$j_3=\rho-j_2$ (this can also be verified directly).} 
\end{proof}

\section{Proof of Proposition \ref{prop:delta}} \label{app:nvar}
\begin{proof}
If $w'_1,w'_2,d$ are integers, then
\begin{equation} \label{eq:bothint}
\nvarsrwt(d) \geq  \sum_{j=0}^{\rho}\big(d - w'_1 j -
w'_2(\rho-j)+1\big) =
(\rho+1)\big(d+1-\frac{\rho}{2}(w'_1+w'_2)\big). 
\end{equation}

Suppose now that $w'_2$ is an integer, while $w'_1$ is not an integer,
say, $w'_1=w'_{10}+1/2$ for some $w'_{10}\in\bbN$.
In this case, if $d$ is an integer and $j$ is odd, then
\begin{multline}
\lfloor d-jw'_1-(\rho-j)w'_2+1\rfloor =  \Big\lfloor
d-jw'_{10}-\frac{j-1}{2}-(\rho-j)w'_2+\frac{1}{2}\Big\rfloor \\
=  d-jw'_{10}-\frac{j-1}{2}-(\rho-j)w'_2
=  d-jw'_1-(\rho-j)w'_2+\frac{1}{2},\label{eq:midway}   
\end{multline} 
while if $j$ is even, the floor function can be simply removed. 
Similarly, if $d$ is not an integer and $j$ is even, then
(\ref{eq:midway}) holds, while if $j$ is odd, then the floor function 
can be removed.
This means that when $w'_1$ is not an integer, we have to subtract $1/2$
either from all odd $j$ or all even $j$ 
of the sum in
(\ref{eq:bothint}). If $\rho$ is odd, we therefore subtract
exactly $(\rho+1)/2$ from
(\ref{eq:bothint}) in any possible case, while if $\rho$ is even, we
subtract at most $(\rho+2)/2$. Hence, 
$$
\nvarsrwt(d)  \geq  \begin{cases}
\displaystyle
(\rho+1)\big(d+\frac{1}{2}-\frac{\rho}{2}(w'_1+w'_2) \big) & \text{if
$\rho$ is odd}\\
\displaystyle
(\rho+1)\big(d+\frac{1}{2}-\frac{\rho}{2}(w'_1+w'_2) \big)-\frac{1}{2}
& \text{if $\rho$ is even}.
\end{cases}
$$
The case where only $w'_2$ is non-integer is handled in a similar
way. 

It follows that when exactly one of $w'_1,w'_2$ is not an integer, for
odd $\rho$, $\mindegrwt(\eneq)$ is not 
above the minimum possible $d\in\frac{1}{2}\bbN$ for which
\begin{equation}\label{eq:mind}
\eneq < (\rho+1)\big(d+\frac{1}{2}-\frac{\rho}{2}(w'_1+w'_2)\big),  
\end{equation}
while for even $\rho$, we have to replace $\eneq$ by $\eneq+1/2$ in
(\ref{eq:mind}). 
As (\ref{eq:mind}) is equivalent to
$$
d>\frac{\eneq}{\rho+1} + \frac{\rho}{2}(w'_1+w'_2)-\frac{1}{2},
$$
(\ref{eq:half}) follows. 
The proof for the case where both $w'_1$ and $w'_2$ are integers
(using (\ref{eq:bothint}) without modifications) is similar.
\end{proof}

\section{Using the monomial order $<_w$ to minimize the
$(1,w+1)$-weighted degree}\label{app:onew}
In the last part of Remark \ref{rem:alga}, it was stated that the
lower of the two weights in (\ref{eq:deltaw}) can be used for both the
even and odd cases. The following proposition makes this statement
precise.

\begin{proposition}
For $w\in \bbZ$ and for monomials $\bs{m}_1,\bs{m}_2\in K[X]^2$
(where $K$ is a field), suppose that $\bs{m}_1<_{w+1}\bs{m}_2$ but
$\bs{m}_1>_{w}\bs{m}_2$. Then
$\wdeg_{w+1}(\bs{m}_1)=\wdeg_{w+1}(\bs{m}_2)$. 
\end{proposition}

\begin{proof}
The assumption implies that $\bs{m}_1,\bs{m}_2$ contain distinct unit
vectors. Hence, there are two cases two consider: (1)
$\bs{m}_1=(X^i,0)$ and $\bs{m}_2=(0,X^j)$, and (2)
$\bs{m}_1=(0,X^i)$ and $\bs{m}_2=(X^j,0)$ ($i,j\in \bbN$). In case
(1), we have $i\leq j+w+1$ and $i>j+w$, which is equivalent to
$i=j+w+1$, as required. In case (2), we have $i+w+1<j$ and $j\leq
i+w$, which is a contradiction. 
\end{proof}


\end{document}